\documentclass[jcp,twocolumn,
nofootinbib]{revtex4-1}
\usepackage{graphicx}
\usepackage{bm}
\usepackage{color}
\usepackage{braket}
\usepackage{amsmath}
\usepackage{amsthm}
\newtheorem{theorem}{Theorem}
\newtheorem{remark}{Remark}
\usepackage{cleveref}
\def\ben{\begin{equation}}
\def\een{\end{equation}}
\def\benu{\begin{enumerate}}
\def\enu{\end{enumerate}}

\begin{document}
\title{Nested Gausslet Basis Sets}
\author{Steven R. White}
\email[]{Author to whom correspondence should be addressed: srwhite@uci.edu}
\affiliation{Department of Physics and Astronomy, University of California, Irvine, CA 92697-4575 USA}
\date{\today}

\author{Michael J. Lindsey}

\affiliation{Department of Mathematics, University of California, Berkeley, CA 94720 USA}

\begin{abstract}
We introduce nested gausslet (NG) bases, an improvement on previous gausslet bases which can treat systems containing atoms with much larger atomic number.  We also introduce pure Gaussian distorted gausslet bases, which allow the Hamiltonian integrals to be performed analytically, as well as hybrid bases in which the gausslets are combined with standard Gaussian-type bases. All these bases feature the diagonal approximation for the electron-electron interactions, so that the Hamiltonian is completely defined by two $N_b\times N_b$ matrices, where $N_b \approx 10^4$ is small enough to permit fast calculations at the Hartree-Fock level. In constructing these bases we have gained new mathematical insight into the construction of one-dimensional diagonal bases. In particular we have proved an important theorem relating four key basis set properties: completeness, orthogonality, zero-moment conditions, and diagonalization of the coordinate operator matrix. We test our basis sets on small systems with a focus on high accuracy, obtaining, for example, an accuracy of $2\times10^{-5}$ Ha for the total Hartree-Fock energy of the neon atom in the complete basis set limit.
\end{abstract}
\maketitle

\section{Introduction}
Electronic structure calculations for molecules and solids are extremely important in science and technology, taking up significant fractions of the world's supercomputing resources. The vast majority of these calculations start with some sort of discretization of space to make the continuum problem finite-dimensional. Most commonly the wavefunction is represented in terms of a basis set for functions of a single space variable.  A significant complication of conventional basis sets for electronic structure, such as atom-centered Gaussian-type orbital bases (GTOs) or plane waves, is the representation of the two-electron Coulomb repulsion as a four-index tensor. Simply storing the full interaction tensor becomes challenging when the number of basis functions exceeds about 1000. Consequently, large-scale calculations use a variety of methods to reduce the $N_b^4$ computational scaling in memory and time that is required merely to define the discretized Hamiltonian, where $N_b$ is the number of basis functions. For example, in the density-fitting/resolution of the identity approach~\cite{dunlap_robust_2000}, an auxiliary basis is constructed for the span of the pair products of basis functions. This reduces the memory scaling for storing the integrals to roughly $N_b^3$, while the computational scaling for downstream use in correlated calculations is highly method-dependent. A more sophisticated compression technique called tensor hypercontraction (THC)~\cite{HohensteinEtAl2012_I,HohensteinEtAl2012_II,ParrishEtAl2013} reduces the memory scaling to roughly $N_b^2$, albeit with a larger preconstant than density fitting, and can be constructed as an interpolative decomposition of the set of pair products of basis functions~\cite{LuYing2015}. A factorization of the interaction tensor into products of $N_b^3$ components plays a role in other approaches~\cite{aquilante_low-cost_2007}, and of course all approaches based on tensor factorization must consider the additional source of method error due to compression. The variety of methods developed to deal with the poor scaling of the interaction tensor demonstrates the practical importance of this issue.

Grid representations potentially yield an even greater reduction in storage: grids naturally allow a diagonal $N_g^2$ representation of the two-electron interaction, where $N_g$ is the number of grid points, but typically $N_g \gg N_b$, largely canceling the improvement. Distorted grids attempt to reduce $N_g$, but there are significant limitations to the amount of distortion, as we shall discuss below~\cite{gygi_real-space_1995}.  

Certain special basis sets can combine the advantages of basis sets and grids. In  particular, a basis set with the {\it diagonal property} permits $N_b^2$ representation of the electron-electron interactions. Moreover, the representation of the electron repulsion integrals is direct, i.e., not presented in a factorized form such as that of density fitting or THC, and the interaction term of the second-quantized Hamiltonian in a diagonal basis only involves density-density contributions. These special properties can improve the scaling of downstream correlated calculations, though scaling details depend on the method. Meanwhile, diagonal basis set approaches enjoy an important advantage over pure finite-difference approaches on grids: as they are basis set methods, it is easy to include extra basis functions, notably atom-centered GTOs, to represent the core electrons and sharp nuclear cusp. 
The oldest basis sets with the diagonal property are grids of sinc functions~\cite{jones_efficient_2016}, but the nonlocality of these functions is a serious limitation. More recently we have developed gausslet bases with the diagonal property~\cite{white_hybrid_2017,white_multisliced_2019,qiu_hybrid_2021}. Gausslets resemble sincs in that they are orthonormal and smooth with a prominent central peak, but they are much more localized in real space.  The diagonal property for gausslets has a clear origin: a gausslet acts like a delta function when integrated against a smooth function, and in this sense it acts like a grid point. Since gausslets are localized, they can be transformed to match their centers to a distorted grid that concentrates functions near nuclei and reduces the $N_b$ required to obtain an accurate approximation.

However, existing methods for distorting grids have significant weaknesses, such as a very limited range of distortion, or the side effect of producing closely spaced functions far from any nuclei. These limitations are most severe for all-electron calculations with larger-$Z$ atoms but already manifest themselves in the first row of the periodic table. By contrast, for systems of hydrogen chains, hybrid gausslet/GTO bases combined with DMRG have provided some of the highest accuracy results available~\cite{white_multisliced_2019}.

Here we introduce a Nested Gausslet (NG) approach which allows arbitrarily large ranges of distortions without artifacts such as closely spaced points far from nuclei. The NG approach is systematically improvable and generates smaller bases than previous methods. This approach has a number of other advantages over previous approaches; for example, one version allows for analytic computation of nearly all integrals defining the Hamiltonian.

In the next section, we review previous distorted grids and introduce nested grids.  The orthogonality pattern of nested grids is discussed in Section III. Nested grids require highly specific properties of the 1D functions used to make them; these properties are explained in Section IV. In Section V we introduce a modification of these bases which allows all integrals to be performed analytically. In Section VI, we present example electronic structure calculations for various small atoms and molecules. In Section VII, we summarize and conclude.

\section{Distorted grids}

Coordinate transformations can be applied either to finite difference grids (FDGs) or diagonal bases (DBs) where each basis function is peaked at a grid point. However, there are also additional techniques for DBs that do not directly apply to FDGs. 
In Fig. 1(a) we show  a very general type of distortion, where the grid is defined by a 3D mapping $\vec r(\vec u)$. In $u$-space the grid is uniform, but in $r$ space it is distorted to put more points near nuclei.  For an FDG, the Hamiltonian is constructed using standard differential approaches for curvilinear coordinates; for a DB, a Jacobian factor accompanies the distorting mapping ensuring orthonormality while approximately maintaining the moments giving the diagonal property. There are several key limitations of this method:  first, the distortion cannot be too large, or else the representation performs poorly, corresponding to the highly distorted squares in Fig. 1(a). Second, we are constrained to use the same number of points along each parallel grid line, so adding sufficiently many points near the nucleus may result in more points than necessary at the periphery of the atom.  These limitations result in rather large grid dimensions, e.g. $96 \times 72 \times 72$ for diatomics~\cite{gygi_real-space_1995}. Third, in the case DBs, the most accurate form of the interaction matrix---the integral diagonal approximation~\cite{white_hybrid_2017}---involves difficult  six dimensional numerical integrals. To our knowledge, the 3D distortions of Fig. 1(a) have only been tested with finite-differences. 

\begin{figure}[t]
    \includegraphics[width=0.96\columnwidth]{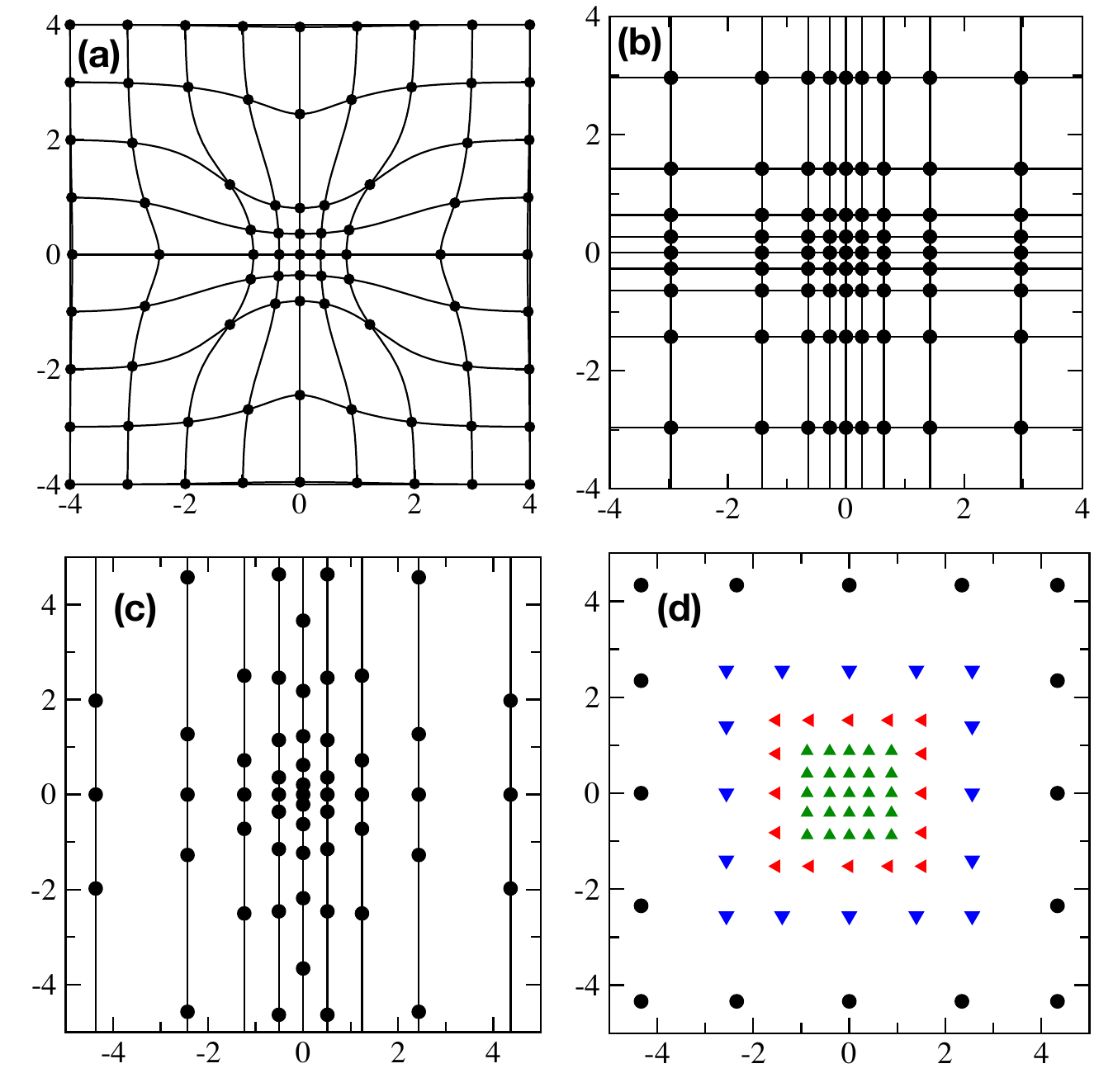}
    \caption{
    Four different approaches to variable resolution for grid/local basis discretizations, in the vicinity of a nucleus at the origin, shown here in 2D. Symbols represent the centers of local basis functions. (a) 3D mapping: a cubic grid is distorted by a 3D mapping function. (b) Products of 1D mappings, where the 3d mapping is a product of mappings for each coordinate direction. (c) Multislicing, where the x direction has one 1D mapping, but then each x point has a separate y mapping. (d) Nested gausslets, introduced here, which can achieve arbitrary resolution at the nucleus while maintaining appropriate spacings in all three directions in all regions. 
\label{Fig1}}
\end{figure}

 Fig.~1(b) shows a simplified mapping that allows for easy numerical integration: the transformation is defined by three 1D functions:  $x(u_x)$, $y(u_y)$, and $z(u_z)$, applied independently to each coordinate.  This approach has been tested on grids~\cite{gygi_real-space_1995} and used for gausslet bases. As long as the grid spacing does not change too rapidly, the diagonal representation of the interactions still works very well, but as mentioned above, the key drawback is that the same number of grid points must be used along all of the parallel grid lines in each direction. 
This method can be used either with FDGs or DBs, but an important advantage of DBs over FDGs is that standard atom-centered Gaussians can be added to the basis, orthogonalizing them with respect to the DB functions, in order to better represent the nuclei without extremely fine grids.  This approach has been used for micro-Hartree accuracy calculations of H$_2$ and He, as well as for Hartree-Fock calculations on H$_{10}$.

Fig. 1(c) depicts multislicing~\cite{white_multisliced_2019}, in which one coordinate mapping is used for the $x$-axis but each line of $y$-functions for fixed $x$ uses a different ($x$-dependent) $y$-mapping. In turn, each line of $z$ functions for fixed $x,y$ (not shown) is built from an $(x,y)$-dependent $z$-mapping. This method is appropriate only for DBs.   In the 2D figure, functions on different vertical lines are orthogonal because the $x$-functions are orthogonal; the different $y$-distortions can be completely different without harming the diagonal approximation. Multislicing eliminates the unwanted close spacing of functions along some coordinate directions (e.g. near (4,0) in Fig. 1(c)) but does not fix the problem in other directions (near (0,4)). Moreover, any rotational symmetry in the system is destroyed by multislicing. Nonetheless, it is more efficient than coordinate-product mapping (Fig. 1(b)), and it has allowed some of the most accurate fully interacting calculations on H$_{10}$ when coupled with DMRG.  

The Nested Gausslet (NG) approach introduced here is  illustrated in Fig. 1(d).  This approach produces orthonormal basis functions organized in shells, shown by the different colors in Fig. 1(d). The arrangement of function centers does not show any of the artifacts along coordinate directions of the earlier approaches, and it can thus produce more compact basis sets.  The shells are still rectangular along the coordinate directions, which in the case of a single atom is inferior to a spherical arrangement. However, the NG shell arrangement can be applied more generally than to a single atom, and the basis functions are composed of coordinate products of functions $f(x) g(y) h(z)$, making numerical evaluation of the integrals very fast. In principle, the basis can be extended to large atomic number $Z$ without an excessive number of functions. The detailed construction of the bases is described below. 

The NG approach allows us not only to zoom into the nucleus efficiently, but also to extend the basis with atom-centered Gaussians, as desired, yielding a hybrid approach. The added Gaussians can reduce the number of points in each shell  without significantly harming the diagonal approximation. 

While multislicing can be performed without significant constraints on the nature of the 1D distortions, the NG approach requires one to construct matching sets of 1D bases  with very special properties. At first glance, the desiderata seem very hard to achieve with only a small number of functions per shell. However, we have found a simple approach to constructing these matching 1D bases. This method is developed in the next few sections. 

\section{Orthogonality pattern for Nested Gausslets}

\begin{figure}[t]
    \includegraphics[width=0.96\columnwidth]{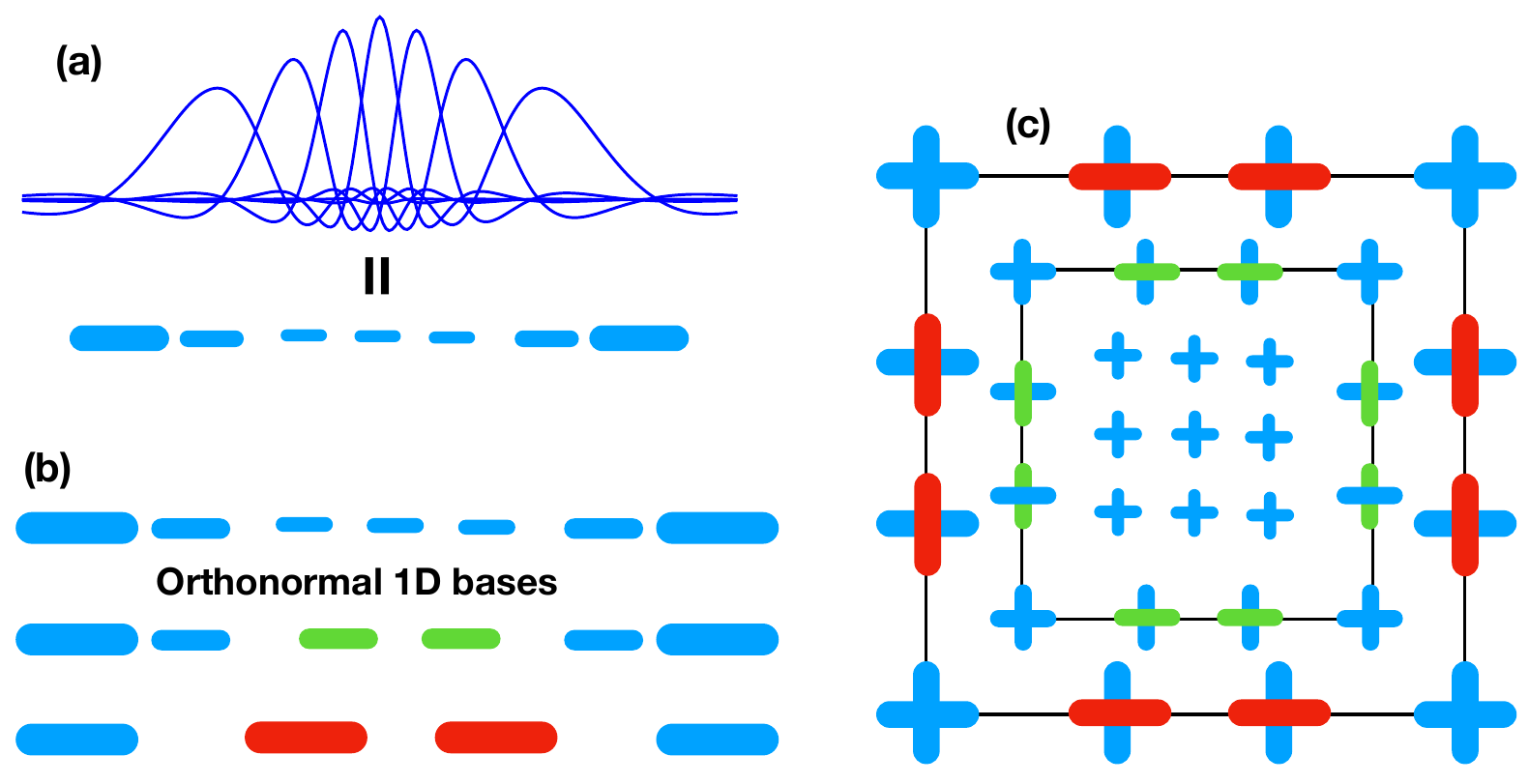}
    \caption{
    (a) Diagrammatic notation for the 1D basis functions used: the distorted 1D basis set is indicated by the array of blue bars. (b) Three orthonormal basis sets used for this construction. The blue functions in the different sets are identical; green and red functions, called {\it side functions}, are orthogonal to the blue functions in the same set. The top set is called the \emph{backbone}. 
    (c) The corresponding nested gausslet layout in 2D. In this simple example the same basis sets are used in both the $x$ and $y$ directions. Each cross is the product of an $x$ and $y$ function, $f(x) g(y)$, with the horizontal and vertical bar indicating $f$ and $g$, respectively. Each 2D function has at least one (blue) backbone function. The 2D functions on different shells (squares) are orthogonal because the backbone functions are orthogonal; on the same shell, they are orthogonal because of the orthogonality of the side bases to themselves and to the outer backbone functions. 
\label{Fig2}}
\end{figure}
We focus first on ensuring that the NG basis is orthogonal, deferring consideration of diagonality to the next section. The layout that ensures orthogonality in shown in Fig. 2. The first step in constructing the basis is to create \emph{backbones} for the $x$, $y$, and $z$ directions. The backbones are distorted 1D gausslet bases which could, e.g., be used as the univariate bases in Fig. 1(b), resolving all length scales that we desired to represent.

We let $n_s$ denote a shell size, i.e., the number of functions along the side of a shell; in Fig. 2, $n_s=4$. We next construct $n_s-2$ orthonormal \emph{side functions} for each shell and each coordinate direction. Each set of side functions replaces an interval of functions on the backbone with a smaller set of functions, and these functions are required to be orthogonal to all the outer backbone functions.  As long as these criteria are met, the whole basis will be orthogonal. The rest of the procedure for constructing the side functions is described in the next section.

Regarding completeness: the nested  construction is related to the coordinate-product mapping of Fig. 1(b): one way to get the side bases is by contraction over the inner sites of the
backbone. Fig. 1(b) can be considered as the limit of no contraction. Contractions which eliminate degrees of freedom should be possible because the backbone must represent the core, but the outer sides are far from the core where the wavefunctions are smoother.  

Generalization to 3D is straightforward:  the shells are cubes, and a basis function on a face but not on an edge is the product of two side functions and one backbone function. The inner side functions form an $(n_s-2)\times(n_s-2)$ grid. We call this construction \emph{singly-nested}, since only the shells are nested, but the faces are treated with rectangular grids. 

A \emph{doubly nested} version is also possible in 3D: each face of a shell would look like Fig. 2(c), with its own nested square shells. This generalization allows for a modest reduction in the number of functions. In the tests below we have used the simpler singly-nested approach.

\begin{figure}[t]
\includegraphics[width=0.99\columnwidth]{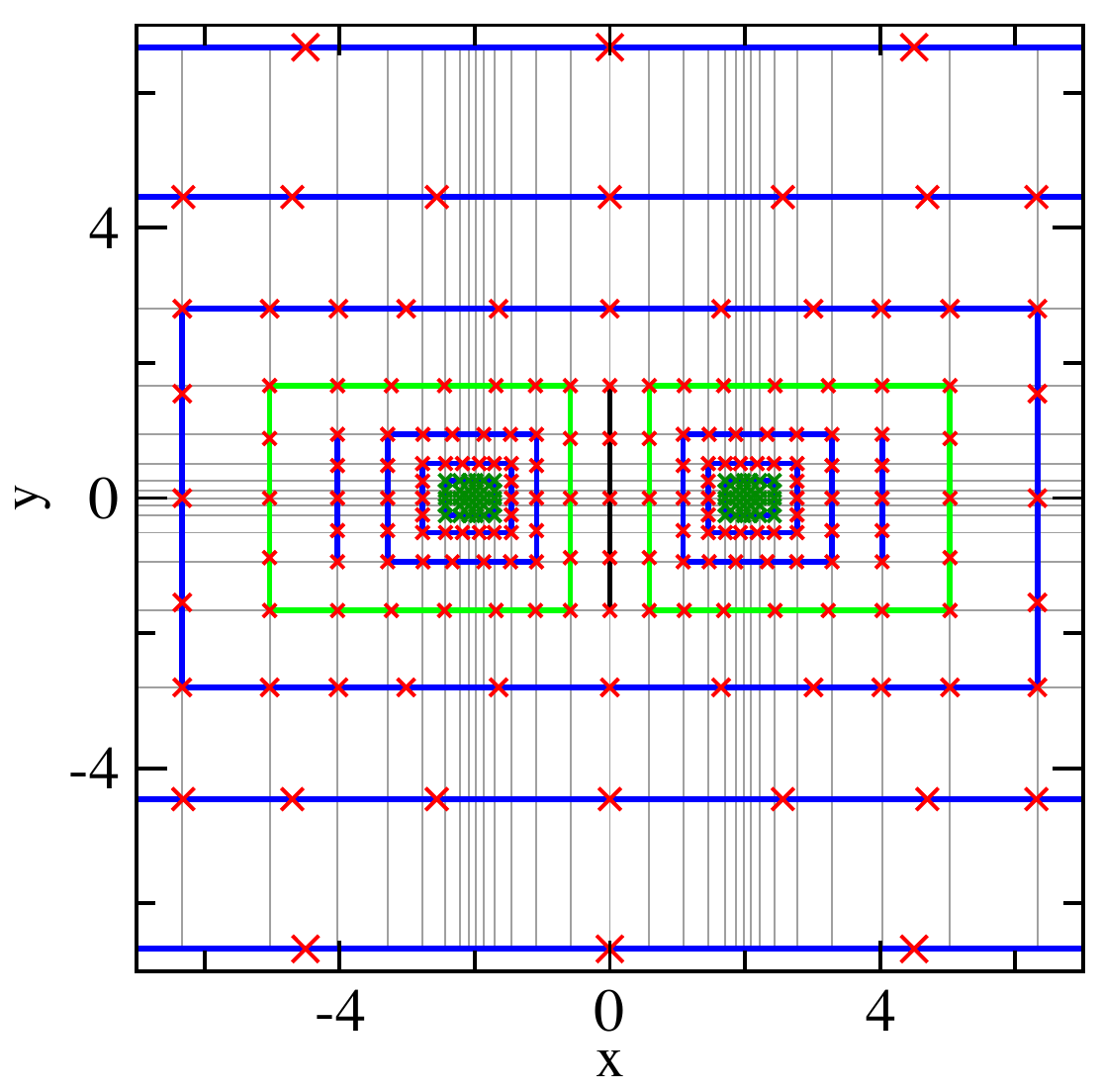}
    \caption{
    Basis function locations for a slice near $z=0$ for Be$_2$ with separation $R=4$ and shell size $n_s=5$.  The locations of the nuclei are at $(\pm 2,0)$. The thin gray lines represent the basis function locations of the $x$ and $y$ backbones. The red X's represent basis function locations. Each red X lies on at least one gray line. The blue lines represent outer shells.  As we progress inward, once the aspect ratio of an outer shell exceeds the number of atoms, the shell is replaced by individual shells for each atom (green rectangles). To maintain reflection symmetry, a flat shell one point wide in the $x$-direction is placed between the atoms (black line). Smaller shells zoom in on each nucleus. As we move toward the nuclei, when the number of gray lines intersecting the shell reaches $n_s$, then the innermost points form a grid based on the backbone functions (dark green region). 
\label{fig:splitting}}
\end{figure}

For more than one atom, a number of generalized constructions are possible.
For this paper, we restrict consideration to the relatively simple case of diatomics and linear chains, where we have developed a fairly simple heuristic prescription.  In Fig. \ref{fig:splitting} we show the arrangements for a typical treatment of Be$_2$ at a separation of $R=4$ a.u. The outer shells surround the whole molecule, while the inner shells switch to being centered on individual nuclei.   The splitting of shells is important for larger $Z$ atoms, but depending
on $R$ and the desired core resolution, the splitting may not occur. 
To maintain reflection symmetry of the basis in the $x$ direction, a flat shell is inserted at the splitting point $x=0$.  Additional flat shells may be inserted at the edges of inner boxes to try to keep the shells as isotropic (i.e., cubic) as possible, as at $x=\pm 4$.

The construction of side functions for these generalized scenarios is not much more complicated than for the single atom case as outlined in the next section. The only slight complication is that for the outer shells, we need to use an effective $n_s$ for the $x$-direction which is larger than the $n_s$ for the $y$ and $z$ directions on that shell.  For example, the smallest (blue) outer shell in   Fig. \ref{fig:splitting} has 11 functions along the $x$ direction.  We choose this number to make the spacing of functions nearly the same in the $x$ and $y$ directions. 


\section{Diagonal 1D bases}
We enumerate several properties that we require of a 1D basis $\{S_i (x) \}$ in order to construct a diagonal approximation.

The first property is {\it completeness} (C). We say that a basis of 1D functions satisfies property C if over the range of interest all polynomials of degree at most $p$ lie in its span. For $p$ sufficiently large, this property guarantees that the basis is rich enough to represent the target wavefunction.

Second is \emph{orthonormality} (O):
\begin{equation}
\label{eq:o}
\int S_i(x) \, S_j(x) \, dx = \delta_{ij},
\end{equation}
or more loosely, as shall be clear from context, orthogonality.

Third is the \emph{moment} property (M), which requires that the basis functions behave like delta functions when integrated against low-order polynomials up to some order $q$. Concretely, we require:
\begin{equation}
\int S_i(x) \, (x-x_i)^m \, dx = w_i \, \delta_{m0} 
\label{eq:mom}
\end{equation}
for a function center $x_i$ and $m = 0, \ldots, q$. Here $w_i := \int S_i (x) \, dx$ is called the weight of the basis function.

Since the zeroth-order moment property is always guaranteed by proper definition of the weight $w_i$, it is equivalent to require
\begin{equation}
\int S_i(x) \, (x-x_i) \, P(x) \, dx = 0
\label{eq:mom2}
\end{equation}
for all polynomials $P$ of order at most $q-1$.

Daubechies compact wavelet-scaling functions typically satisfy properties CO, but not M. (In our discussion of `wavelets' as candidates for basis functions with the diagonal property, we actually mean the \emph{scaling function}, not the wavelet itself.) One particular type of compact wavelet, the coiflet, satisfies COM. Meanwhile, sinc functions satisfy CO, but higher order moments diverge due to their algebraic (i.e., subexponential) decay. Many non-compact wavelets (e.g., Meyer scaling functions) satisfy COM. Some quantitative comparisons are offered in Appendix C.


Finally, we highlight a fourth important property: \emph{X-diagonalization} (X). This property holds if the functions diagonalize the $x$-coordinate operator, i.e., if 
\begin{equation}
\label{eq:x}
\int S_i(x) \, x \,  S_j(x)  \, dx = \tilde x_i  \, \delta_{ij}, 
\end{equation}
where  $\tilde x_i$ is a function center which may be different from the moment center $x_i$.
X-diagonalization is important in the construction of discrete variable representations~\cite{light_discrete-variable_2000}, which are related to our diagonal bases. 

Coiflets satisfy COM but not X. However, the ternary wavelets of Evenbly and White (EW)~\cite{evenbly_representation_2018} satisfy COMX.  While the satisfaction of COM was essential to their construction, the X property comes about from their even symmetry. Indeed, since $S_i(x)$ is even about $x_i$ and $S_j(x)$ is even about $x_j$, their product is even about the midpoint $(x_i + x_j)/2$.  Thus $\int S_i(x) \, (x-x_m) \, S_j(x) \, dx = 0$.  Since $S_i$ and $S_j$ are orthogonal for $i \neq j $, it follows that $\int S_i(x)\,  x \, S_j(x) \, dx = 0$ when $i \neq j$, and property X is guaranteed. The non-evenness of coiflets interferes with property X. This non-evenness is tied closely to their compactness coupled with their usual binary construction. Thus the EW ternary wavelets are very unusual in being COMX and compact.

Gausslets are based on EW ternary wavelets~\cite{white_hybrid_2017}, yet they satisfy some properties only inexactly up
to high precision. The lack of exactness comes from their simple
analytic presentation as linear combinations of Gaussians. Their completeness, for example, holds up to a tolerance of roughly $10^{-8}$, so that when they are used in a variational calculation one still can obtain energies accurate to order $10^{-16}$, i.e., floating point double precision. Thus for practical purposes, gausslets are COMX. They are also for practical purposes compact, since their non-compact tails decay as $e^{-\gamma x^2}$ where $\gamma = \frac{9}{2}$. 

We now prove an important theorem relating C, O, M and X.

\begin{theorem}
\label{thm:comx}
{\rm (COMX Theorem)} Suppose that the collection of functions satisfies properties C, O, and X and in particular satisfies C with polynomial completeness up to order $p$. Then in fact the collection also satisfies property M up to order $q = p+1$ (where the function centers $x_i$ are chosen to be the same as the function centers  $\tilde{x}_i$ in property X).
\end{theorem}

Before giving the proof, we pause to include some remarks on the interpretation.

It is instructive to think of properties C and M as properties of the \emph{subspace} spanned by our basis of functions. The reason is that, given such a subspace, diagonalizing the X operator on this subspace always yields the unique basis of functions that span this subspace and in addition satisfy properties O and X. Therefore, O and X are guaranteed for free, and the orders $p$ and $q$ to which C and M are respectively satisfied by this basis can be thought of as intrinsic properties of the subspace.

From this point of view, the COMX Theorem 
then indicates that to satisfy COMX, it actually suffices to construct a subspace that is polynomially complete to some desired order (C). The moment property M then also comes for free via diagonalization of the X operator.

Note with caution, however, that COMX can be satisfied by a basis of functions that are not very localized. Indeed, from the subspace point of view, there exist subspaces with excellent completeness properties that are impossible to localize. For example, one could add a very high frequency Fourier mode to a basis that is already localizable. The resulting basis is then impossible to localize without altering the subspace, even though diagonalization of the X operator on this subspace would produce a basis satisfying COMX.

\begin{proof}[Proof of Theorem \ref{thm:comx}]
Let $\{ S_i \}$ satisfy C, O, and X, with polynomial completeness (C) up to order $p$. We will show that \eqref{eq:mom2} holds with $q = p+1$. As such let $P$ be a polynomial of order at most $p$.

Linearly combining equations \eqref{eq:o} and \eqref{eq:x}, in which we take $x_i = \tilde{x}_i$, we obtain 
\begin{equation}
\label{eq:combineCX}
\int S_i (x)\,  (x - x_i ) \,  S_j (x) \, dx = 0
\end{equation}
for all $i,j$.

Now fix $i$. Since the collection $\{ S_j (x) \}$ is polynomially complete to order $p$, there exists a linear combination of the $S_j$ that matches the target polynomial $P$. Accordingly, taking an appropriate linear combination of \eqref{eq:combineCX} over the index $j$, we obtain the equality
\begin{equation}
\int   S_i (x) \, (x - x_i ) \, P(x) \, dx = 0,
\end{equation}
as was to be shown.
\end{proof}
\begin{remark}
Note that the argument in the proof assumes that polynomial completeness (C) holds \emph{exactly} up to order $p$. In practice, C may hold only up to some small numerical tolerance, and likewise the moment property M will then hold only up to a small tolerance. It is difficult to formulate a clean notion of numerical tolerance for C since polynomials grow without bound, so instead we outline a small mutation of the proof that controls the error in property M resulting from imperfect satisfaction of property C.

For any fixed $i$ and polynomial $P(x)$ of order at most $p$, let $R(x)$ denote the remainder of some fit of $P(x)$ with a linear combination of the $ S_j (x) $. Then equation~\eqref{eq:mom} for property M holds with error
\[
\left\vert \int S_i (x) \, (x-x_i ) \,  R(x) \, dx \right\vert  \leq \sigma_i \,  \sqrt{\int R(x)^2 \, \vert S_i(x) \vert \,dx},
\]
where $\sigma_i := \sqrt{ \int (x-x_i)^2 \, \vert S_i (x) \vert \, dx}$ and the inequality holds by Cauchy-Schwarz. Hence to control the theoretical worst-case error in property M it makes sense to obtain the fit in the weighted least squares sense with respect to the weight function $\vert S_i (x) \vert$. In practice, the violation of M can simply be estimated numerically. 
\end{remark}

\begin{figure}[t]
\includegraphics[width=0.75\columnwidth]{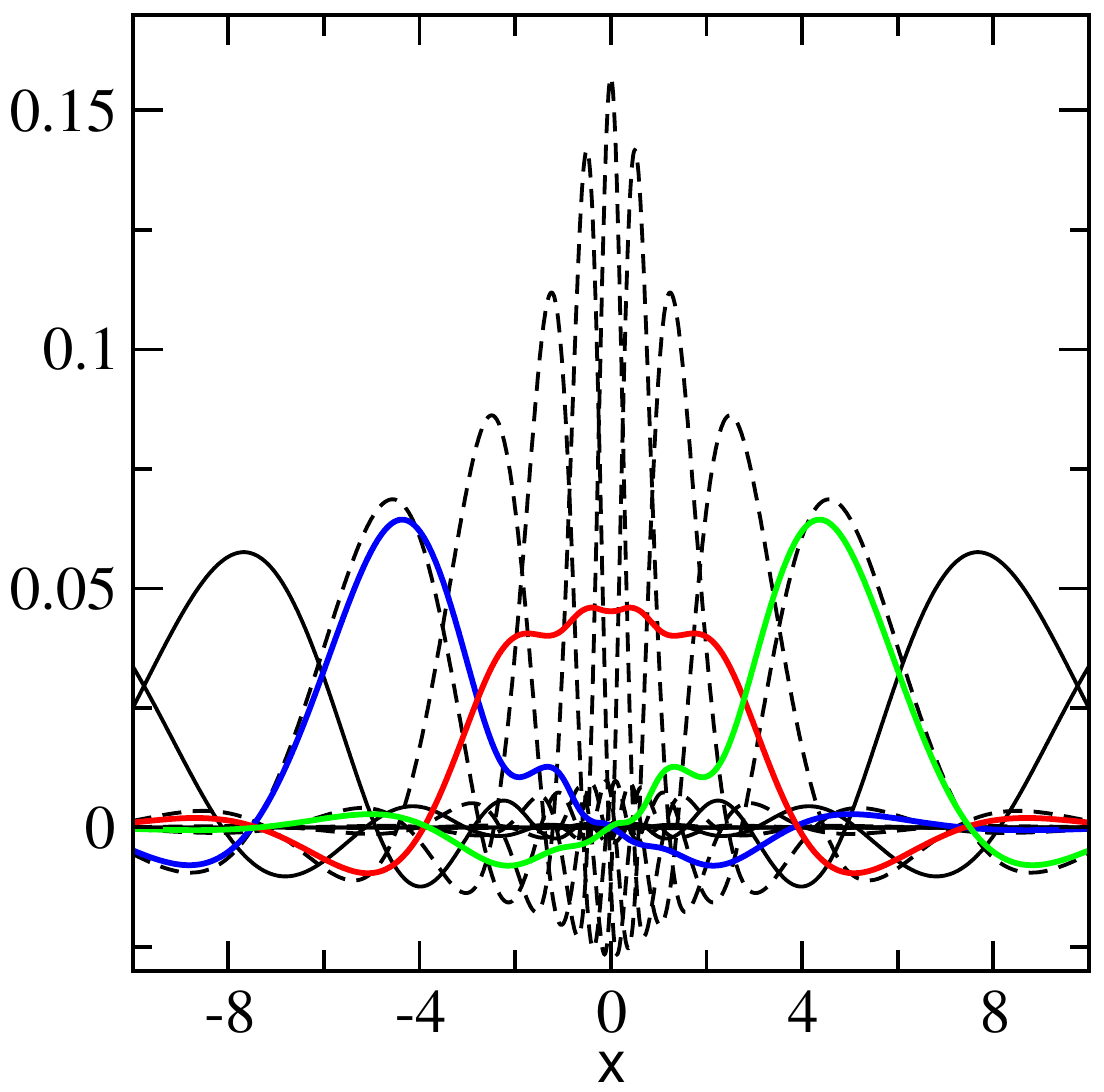}
    \caption{
    Backbone and side basis with $n_s=5$. The backbone, shown in black (solid and dashed) is a distorted set of gausslets using the sinh coordinate transformation with $a=s=0.7$. The side basis (colored lines) is constructed as a contraction of the inner backbone (dashed lines).   The full side basis is composed of the solid black lines, together with the colored lines. 
    It is orthonormal and complete up to quadratic polynomials, with moments up to cubic vanishing. The corresponding shell in the nested gausslet construction uses the inner $n_s=5$ functions of the full side basis. 
\label{fig:backboneside}}
\end{figure}

A first important practical consequence of the theorem is that it makes it easy to construct the side functions.  Assume that the backbone $\{b_i (x) \}$ is COMX.
The side basis can be constructed by replacing an interval of functions from the backbone with a smaller set of $n_s-2$ functions $\{S_i(x)\}$. In particular, we can choose the $S_i$ as linear combinations of the $b_i$, i.e., $S_i = \sum_j A_{ij} b_j$, where $AA^{\top}=I$ guarantees orthonormality. More specifically, for optimal completeness, we project the monomials $x^j$ onto the space spanned by the $b_i$, yielding  the functions  $\tilde S_j = \sum_k b_k \langle b_k| x^j \rangle $ for $j=0,\ldots n_s$. Then we orthonormalize the $\tilde S_j$ and diagonalize the X operator projected to this basis. The side functions $S_i$ are recovered as the resulting eigenfunctions. An example of the resulting side basis functions is shown in Fig. \ref{fig:backboneside}.

Since the $b_i$ are orthogonal to the outer backbone functions, so are the $s_i$. Since the $x$ matrix for the backbone is diagonal, diagonalizing within the $s_i$ space makes the whole sequence X. The side basis is complete up to polynomials of order $n_s-3$. Satisfying COX,
the $s_i$ also satisfy M up to order $n_s-2$. 

We comment that during this construction (specifically, within the expression for $\tilde{S}_j$), it is not necessary to compute the high-order polynomial overlaps $\phi_{kj} =\langle b_k|x^j\rangle$ explicitly, and in fact a direct approach may be subject to rounding errors.

Instead we adopt the following approach. To begin we set some notation and assumptions. Let $P$ and $Q$ project onto the backbone space within and outside the interval, respectively, so in particular $PQ = QP = 0$. In the following, for emphasis and clarity we will use the notation $[f(x)]$ to indicate the diagonal operator that performs pointwise multiplication by $f(x)$, and we will use the notation $(f(x))$ to denote the function itself. In particular, $(1)$ denotes the constant function with value 1. Normally we would omit such special notation when the meaning is clear. Because of the X property of the backbone, the $P$ and $Q$ spaces block diagonalize $[x]$, and therefore $P[x]Q = 0$. We assume that the backbone space is polynomially complete at least up to order $j-1$. Since the projector onto the backbone space is $P+Q$, this means that $(P+Q)(x^k) = (x^k)$ for $k=0,\ldots,j-1$.

Then observe that $\phi_{kj} = \langle P b_k \vert (x^{j}) \rangle = \langle b_k \vert P (x^{j}) \rangle$. Moreover, we can write $(x^{j}) = [x] (x^{j-1}) $, so $\phi_{kj} = \langle b_k \vert P [x] (x^{j-1}) \rangle$. But then we can substitute the identity $(x^{j-1}) = (P+Q)(x^{j-1})$ and use the fact that $P [x] Q = 0$ to obtain $\phi_{kj} = \langle b_k \vert P [x] P (x^{j-1}) \rangle$. Continuing in this fashion, we obtain  $\phi_{kj}=\langle b_k|P [x] P [x] \cdots [x] P (1) \rangle$.

Now note that $P=\sum_k |b_k\rangle\langle b_k|$, and let $X = (X_{kl}) = (\langle b_k \vert [x] \vert b_l \rangle )$ denote the matrix of the $[x]$ operator within the $b_k$ basis. Then it follows that the $j$-th column $\phi_j$ of $(\phi_{kj})$ can be written as the matrix-vector product $\phi_j = X^j \xi$, where $\xi$ is the vector with entries $\xi_m = \langle b_m \vert 1 \rangle$.

Therefore the required polynomial functions are spanned by the Krylov space of the matrix $X$ acting on the initial vector $\xi$.  It follows that we can use the Lanczos procedure to directly construct an orthogonal basis for the side functions for an interval. Then a final diagonalization gives the COMX side basis.

\section{Pure Gaussian Distorted gausslet bases}

A uniform gausslet basis, since it is formed from contractions over Gaussians, allows all standard integrals to be computed analytically. Adding a coordinate transformation eliminates this property.  The smoothness of gausslets and the transformations we use nevertheless allows for effective numerical quadrature, and previous use of gausslets have been based on numerical quadrature, coupled with the decomposition of $1/r$ into a sum of Gaussians as discussed in Appendix B. In this section we introduce an approach which restores the analytic integration. 

A 1D gausslet with standard unit spacing is defined as a contraction over a uniform array of Gaussians with width $1/3$ and spacing $1/3$, i.e.,
\begin{equation}
G(x) = \sum_i c_i e^{-\frac{9}{2}(x-i/3)^2},
\end{equation}
where the sum runs over about 100 nonzero values of $c_i$. A coordinate transformation of the gausslet acts on the underlying Gaussians, leaving the contraction the same. Concretely, given a distortion $u(x)$ with density $\rho(x)$ and inverse $x(u)$ (see Appendix A) we can define the \emph{distorted gausslet (DG)} as 
\begin{equation}
\sum_i c_i e^{-\frac{9}{2}(u(x)-i/3)^2} \sqrt{\rho(x)}.
\end{equation}
The functions so defined form an orthogonal diagonal basis. 

Our goal is to replace the DG basis with a basis that permits analytic integrations. Now because the coordinate transformation we use is slowly varying compared to the gausslet spacing, and because the Gaussians are spaced three times as finely, the individual Gaussians comprising the gausslet are mapped under this transformation to functions which still resemble Gaussians. However, the locations, widths, and amplitudes must be appropriately transformed. The center of the $i$-th transformed Gaussian is given by $x_i = x(i/3)$, 
its nominal width  is $W_i = 1/[3 \rho(x_i)]$, and its amplitude is increased by a factor of $\sqrt{\rho(x_i)}$.
 
In the \emph{pure Gaussian distorted gausslet (PGDG)} basis, we replace the individual distorted Gaussians with undistorted Gaussians defined using these transformed parameters:  
\begin{equation}
\tilde G(x) = \sum_i c_i \sqrt{\rho(x_i)} e^{-\frac{1}{2 W_i^2}(x-x_i)^2}.
\end{equation}
This replacement destroys the exact orthonormality of the basis, i.e., property O. Therefore, as discussed in the previous section, we orthogonalize and X-diagonalize the $\tilde G$ functions. This maneuver exactly restores properties OX and approximately restores the CM properties. We find that the new functions are very similar to the originals, as we illustrate in Fig. \ref{fig:pgdg}, which shows how similar a distorted gausslet and its corresponding pure Gaussian analog are for a typical distortion. Moreover, we find that the subspace spanned by the $\tilde G$ appears to be very similar to that of the distorted gausslet basis. This point is illustrated in Fig. \ref{fig:fitstogau}, which demonstrates the fitting of Gaussian type functions 
using both DG and PGDG bases. The results are almost identical to a surprising degree. Note that the distortion itself makes the CM properties only approximate for distorted gausslets; the
replacement with undistorted Gaussians might be expected to make the CM errors bigger, but in practice any such effect appears to be insignificant. 

\begin{figure}[t]
\includegraphics[width=0.6\columnwidth]{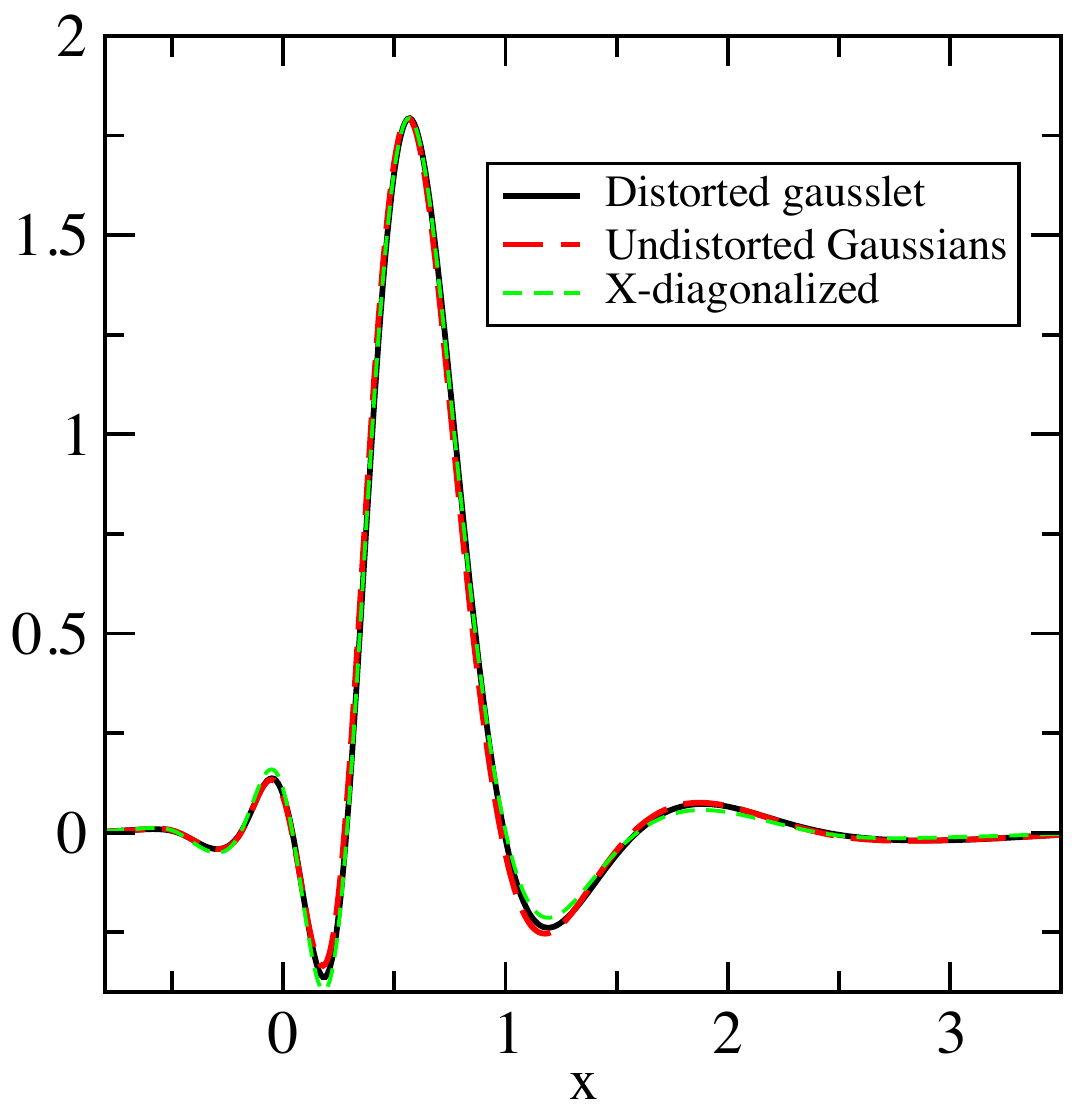}
    \caption{
For a typical coordinate mapping for a nucleus at the origin, the black line shows a standard distorted gausslet. The red dashed curve shows the function with the underlying distorted Gaussians replaced by undistorted Gaussians.  The green dashed curve shows the function after the new set of functions is orthogonalized and X-diagonalized.  The functions are all very similar.
\label{fig:pgdg}}
\end{figure}

\begin{figure}[t]
\includegraphics[width=0.7\columnwidth]{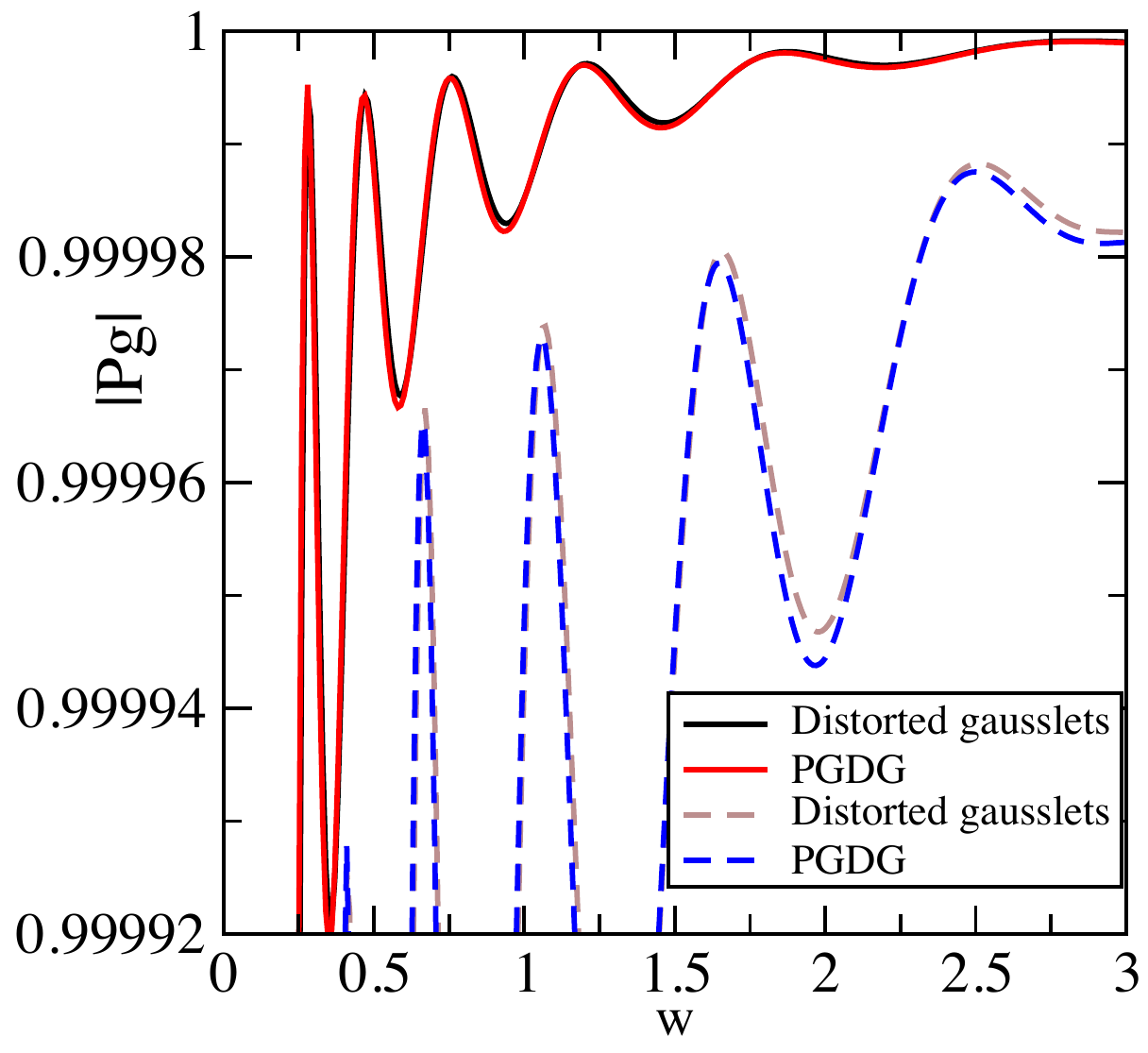}
    \caption{
The upper two curves show fits to a normalized Gaussian with the specified width $w$, $\exp(-\frac{1}{2}x^2/w^2)$, for a pure distorted gausslet basis and the corresponding PGDG basis. The fit is defined as $|P g|$, where $g$ is the Gaussian, $P$ is a projector into the basis, and an L2 norm is taken. A perfect fit would give 1. The lower two curves show similar fits for $x$-Gaussians,  $x \exp(-\frac{1}{2} x^2/w^2)$.  The close similarity of the curves shows that the span of the DG and PGDG bases are nearly identical, justifying the use of the PGDG bases, which have analytic integrals.
\label{fig:fitstogau}}
\end{figure}

Note that since the side functions are linear combinations of the backbone functions (for each direction $x$, $y$, or $z$), knowing all integrals involving the backbone functions allows us to produce all 1D integrals via matrix multiplications. The gausslets comprising the backbone in a particular direction share the same Gaussians, so it suffices to obtain all integrals for this one set of Gaussians of size $N_g$ (for each direction). Typically $N_g \sim 200$. The number of 1D integrals behaves as $N_g^2$ for the kinetic energy and $N_g^2 M$ for the Coulombic cases, where $M \sim 100$ is the number of Gaussians used to decompose $1/r$ (see Appendix B). This modest number ($\sim 10^7$) of analytic  1D integrals can be computed very quickly, and their 1D linear transformations to represent backbone and side functions reduce to fast matrix-matrix multiplications. 

\begin{figure}[t]
\includegraphics[width=0.8\columnwidth]{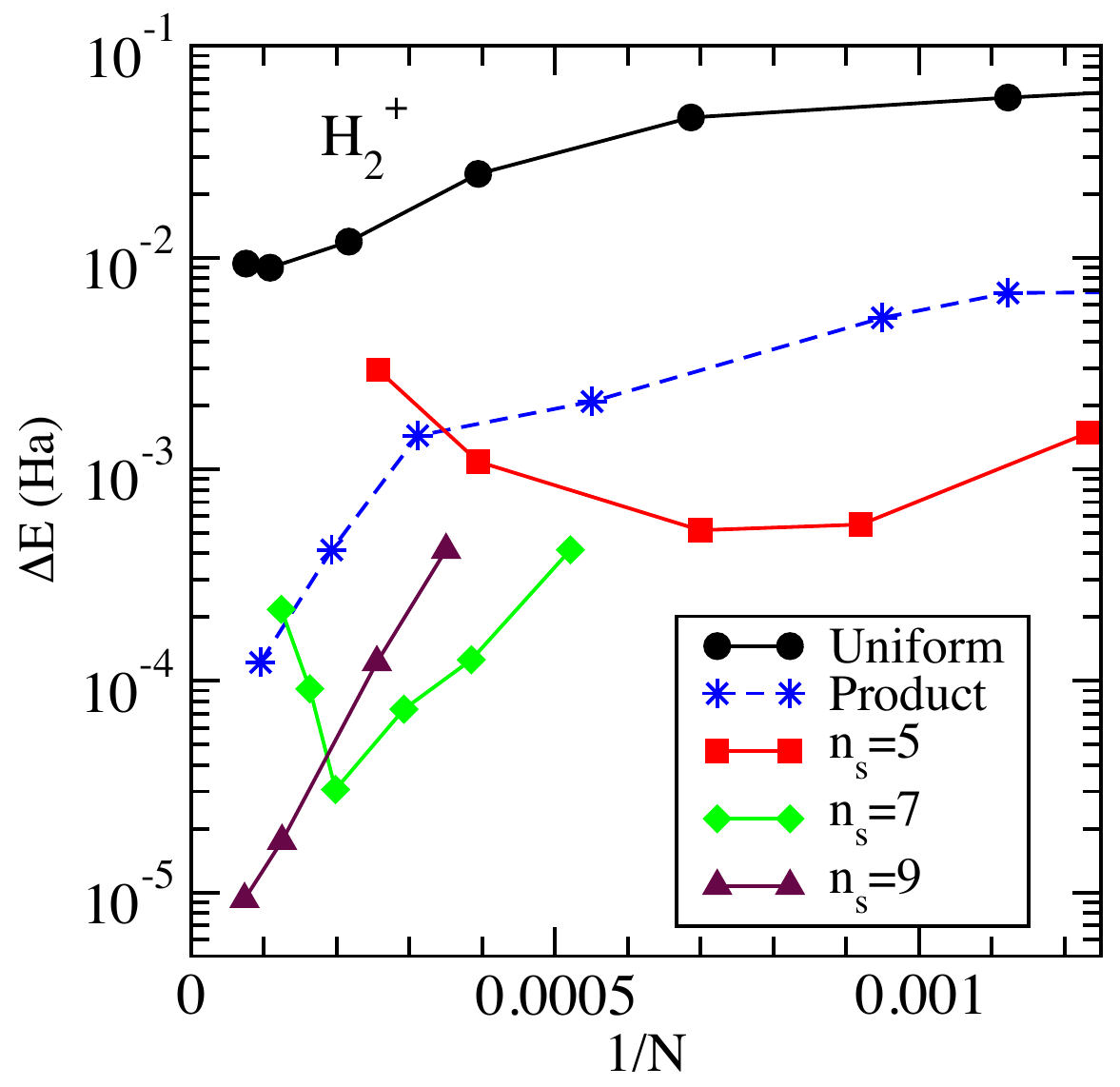}
    \caption{
Comparison of results for H$_2$ positive ion at separation $R=2$ a.u. for the non-hybrid method, i.e., without any additional 3D Gaussians from a standard GTO basis. In each case, the error in the energy relative to the precise value of  -0.60263462~\cite{scott_new_2006} is shown as a function of the inverse of the number of basis functions. The curve labeled Uniform represents an equal spaced array of gausslets. The curve labeled Product is based on the product coordinate transformation of Fig. 1(b). The other three curves are based on nested PGDG bases with the given side dimension $n_s$. One moves to the left along these curves as the discretization of the backbone is refined.
\label{fig:figcomp}}
\end{figure}

Our basis consiting of a total of $N_b$ 3D functions can be formed from products of 1D functions of the form $f(x) g(y) h(z)$ . Since we use the diagonal approximation for the electron-interaction, not only are the nuclear potential and kinetic energy represented by $N_b \times N_b$ matrices, but the electron-electron interaction is as well. The final construction of the two Coulombic matrices scales as $N_b^2 M$, which allows us to treat $N_b \sim \textrm{10,000}$ on a desktop. Note that this is better scaling than the $N_b^3$ required to simply diagonalize the one electron part of the Hamiltonian.

As a first measure of the effectiveness of the PGDG nested gausslet approach, we compare calculations for the hydrogen molecular ion in Fig. \ref{fig:figcomp}.
The nested approach gives nearly an order of magnitude improvement in the energy error over the product coordinate transformation for a fixed number of basis functions. Meanwhile it yields a couple of orders of magnitude improvement over a naive uniform basis, which is in turn expected to perform similarly to a plane wave basis of the same size.  Note that each curve for fixed $n_s$ has a minimum, beyond which increasing the radial accuracy (i.e., more shells) with fixed angular accuracy (i.e., fixed $n_s$) is counterproductive. Note that all of these bases can serve as part of a hybrid Gaussian/gausslet basis, discussed in the next section. Adding some Gaussians improves all the approaches and obscures the difference between them, so for Fig. \ref{fig:figcomp} we omit Gaussians, but in general, we always prefer the hybrid bases.

\section{Hybrid nested gausslet/Gaussian bases}
Atom-centered basis sets, such as the standard Gaussian type orbital bases (GTOs) used throughout quantum chemistry, have impressive completeness for their size which is hard to duplicate directly with NGs. However,
we can incorporate all or part of a GTO basis into an NG basis.  The point of doing this, rather than simply using a GTO basis, is to retain the diagonal form of the interactions. Moreover, compared to GTOs, the NG basis is easier to systematically improve away from the atom centers. How one combines gausslet and GTOs while maintaining diagonal interactions is described in Ref.~\cite{qiu_hybrid_2021} for product coordinate transformations, and no significant changes are needed for nested gausslets. Performing the integrals needed for the hybrid PGDG/GTO approach is particularly convenient: all the integrals in both bases are analytic.

The most important purpose of the GTOs is to resolve the nuclear cusp efficiently, reducing the number of shells needed. They also help supplement the completeness of the NG basis, by, for example, helping restore spherical symmetry near the atom centers. These benefits are largely derived from the low angular momentum functions, and little benefit is seen from adding functions beyond
S and P.  In most cases we add just the S and P functions.

The main idea underlying the hybrid basis construction is to view the gausslets as the main functions and orthogonalize the GTOs to the gausslets, forming Residual Gaussians (RGs).  The RGs have quite low occupancy; they serve as corrections to nuclear cusps and are generally highly oscillatory. While they are important for high-accuracy single-particle properties, namely the kinetic and nuclear potential energies, they are less important for the electron-electron interactions. Therefore simplistic approximations for their interaction terms are good enough. The simplest approximation, called Gaussian–gausslet transfer (GGT), replaces each RG by the gausslet closest to the center of the RG and simply uses the gausslet-gausslet integrals, which one must compute anyway. Although some other approximations are somewhat better than GGT,  the error discrepancies are small compared to other errors, such as the overall diagonal approximation error incurred just by the gausslets. 

Here we introduce and utilize an alternative to GGT which is about as simple, but which has smaller errors. A collection of GTOs $G$ is orthogonalized to the gausslets and symmetrically orthonormalized to yield corresponding RGs $\hat G$. Now $\hat G$ and $G$ are nonzero in the same general area, although $\hat G$ is largely high-momentum and somewhat more extended than $G$. However, $G^2$ and $\hat G^2$ are more similar to each other, in that both have low-momentum components which contribute strongly to the interactions. These low-momentum components can be represented by the gausslets $\{g\}$.  In the \emph{density transfer approximation (DTA)} we approximate the low momentum components of $\hat G^2$ by $N_G \sum_g g^2 \langle g|G \rangle^2$, where $N_G = 1/\sum_g \langle g|G\rangle^2$ is a normalization factor which lends the sum an interpretation as a weighted average. In this expression we have approximated a double sum over $g$ and $g'$ as a diagonal sum over $g$, since the off-diagonal products $g g'$ have no zero momentum components due to orthogonality. 
If our diagonal two-electron gausslet-gausslet interaction is described by a matrix $V_{gg'}$, then DTA corresponds to using  $V_{g\hat G} = N_G \sum_{g'} V_{gg'} \langle g'|G\rangle ^2 $ and $V_{\hat G\hat G'} = N_G N_{G'}\sum_{g'g''} V_{g'g''} \langle g'|G\rangle ^2 \langle g''|G'\rangle^2$. 

Comparing DTA to GGT, we find that DTA is better. For example, one check on the interaction matrix for an atom is to evaluate the repulsion for two electrons in a hydrogenic 1S orbital versus atomic number $Z$: it should evaluate to $5/8 Z$.  To test this result, we obtain the 1S orbital (within the basis) as the ground state of the one-electron Hamiltonian.  For a typical case, a He atom with $n_s=5$ and a minimal spacing of the gausslet basis at the origin of 0.3, with the $S$ and $P$ functions of the cc-pV6Z basis\cite{Heccpv6z} added to make the hybrid basis, we find $1.25020$ for DTA and $1.25037$ for GGT, versus $1.25$ exact.  We have not compared DTA with the other more complicated approximations in~\cite{qiu_hybrid_2021}.

A variety of corrections were explored in~\cite{qiu_hybrid_2021}, to make the hybrid bases more accurate. Most of these are less important for nested gausslets because of the better core resolution, so we do not include them in our results. The exception is the double-occupancy electron-electron cusp correction, which is relevant for full-CI calculations.
This is a correction to the energy added at the end of a calculation, given by
\begin{equation}
    \Delta E = e_0 \sum_i d_i^\alpha,
\end{equation}
where $d_i = \langle n_{i\uparrow} n_{i\downarrow} \rangle$ is the double occupancy of gausslet $i$.  The parameters were fitted to a number of gausslet two-electron systems and found to be fairly universal, with 
\begin{equation}
e_0 = -0.005078
\label{Eqn:cusp_e0}
\end{equation}
and
\begin{equation}
\alpha = 0.79.
\label{Eqn:cusp_alpha}
\end{equation}
In previous high-accuracy calculations, the cusp correction could reduce errors by about an order of magnitude. We utilize these parameters in our full-CI tests.

\section{Results}
In this section we show results from
hybrid gausslet/GTO bases, where the gausslet is a PGDG. 

\begin{table}[!htb] 
\centering
\begin{tabular}{|c||c|c|c|l|l|l|c|}
\hline 
Basis& $n_s$ &d&$R_b$& E & V & $N_b$\\ 
\hline
DZ &G&&& -0.49928               &              &  5 \\ 
DZ &5&0.4&4& -0.49947               &0.6265    & 326  \\
DZ &5&0.4&6& -0.49953               &0.62470    & 424  \\
TZ &G&&& -0.49981               &              &  14 \\ 
QZ &G&&& -0.499946              &              &  30 \\
QZ &7&0.2&6& -0.499968              &0.624988    & 1446  \\ 
QZ &7&0.2&8& -0.499969              &0.624960    & 1664  \\ 
5Z &G&&& -0.4999945             &              &  55 \\
5Z &9&0.1&6& -0.4999968             &0.625032    & 3834  \\
5Z &9&0.1&8& -0.4999972             &0.6249989    & 4220  \\
6Z &G&&& -0.49999924            &              &  91 \\
6Z &9&0.1&8& -0.49999950            &0.6250011    & 4224  \\
 \hline
\end{tabular}
\caption{Energies from a hybrid PGDG-gausslet/GTO basis for a single hydrogen atom. Also shown are the total energies for a pure Gaussian basis (with the label G in the $n_s$ column). Here the GTO bases are from the cc-pVXZ family\cite{Hccpv6z}, $d$ is the minimum gausslet spacing, and $R_b$ denotes the size of the box containing the gausslets ($2R_b\times2R_b\times2R_b$). $N_b$ is the total basis size. The exact total energy $E$ is -1/2, and the exact interaction energy $V$ between two electrons placed in the exact 1S orbital is -0.625. }       
\label{tab:Henergies}
\end{table}

In Table \ref{tab:Henergies}, we show results for the energy of the hydrogen atom. The table compares pure GTO results with results for hybrid NG/GTO bases of similar accuracy, where the number of gausslets primarily affects the two particle (V) term through the diagonal approximation. Of course, in this case the GTO basis is vastly more efficient, but this very simple example gives an indication of the accuracy of the diagonal approximation which is similar in more complicated systems. Note that the measured interaction energy $V$ is not a true property of a hydrogen atom: it assumes that two electrons exist in the noninteracting single-particle 1S orbital. The diagonal approximation is not variational, and $V$ may be above or below the exact result, but its error decreases very rapidly with $n_s$. In this case we see that chemical accuracy in $V$ is easily obtained with only $n_s=5$, and $\mu H$ accuracy is obtained with $n_s=9$. This high accuracy is obtained from the diagonal approximation despite the singular nature of the Coulomb interaction at short distances.

The size of the gausslet basis is controlled through $n_s$, the minimum spacing $d$, and the box size $R_b$. Note that despite the small size of the pure Gaussian bases, the approximate storage for the Hamiltonian is comparable to that of the hybrid bases, since storage varies as $N_b^4$ in the Gaussian case, versus $N_b^2$ for gausslets (assuming factorized forms are not used). For molecules, the Gaussian basis size grows linearly with the number of atoms, but for gausslets, the growth in basis size is significantly slower, since overlapping regions of atoms share the same gausslets.  However, the most important consideration for systems beyond the hydrogen atom is the slow convergence of the correlation energy with the size of the basis. Indeed, roughly speaking, the basis error in the correlation energy varies as $1/N_b$ for both Gaussians and gausslets, in the high accuracy limit where the electron-electron cusps are being resolved. The larger number of gausslets required to make the diagonal approximation accurate helps to converge the correlation energy more quickly.

\begin{figure}[t]
\includegraphics[width=0.8\columnwidth]{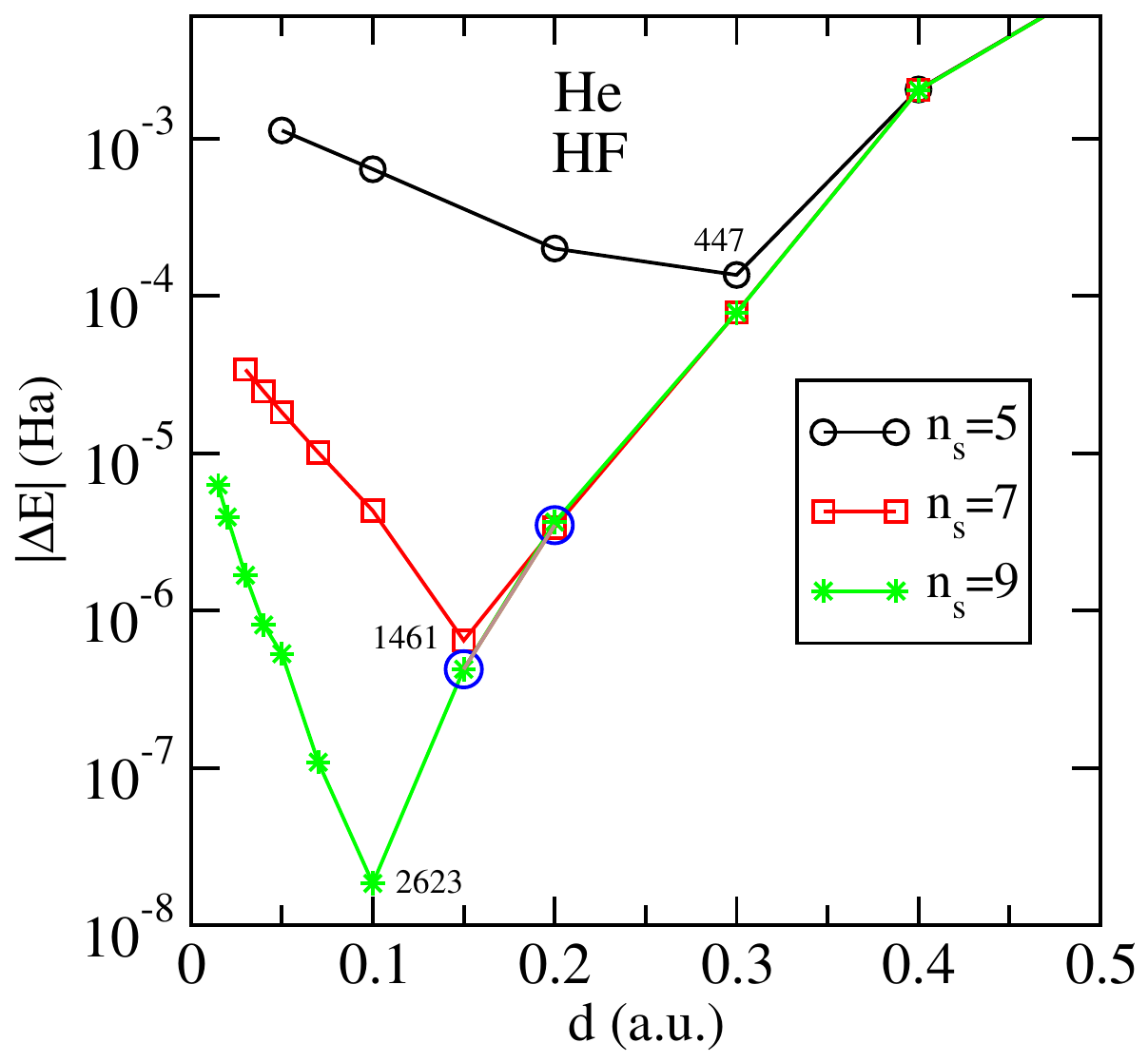}
    \caption{
Hartree-Fock results for a helium atom in a hybrid PGDG-gausslet/GTO basis, versus the minimum gausslet spacing $d$ (in a.u.), for box size $R_b=7$ bohr.  The basis used is the hydrogenic AHGBS-9 basis which only has S Gaussians and is highly accurate for the nuclear cusp. The energies are measured relative to the highly accurate result of Cinal~\cite{cinal_highly_2020}, -2.8616799956122....  The points enclosed in the two blue circles have negative energy errors, which is possible because the diagonal approximation for the interactions is not variational. The number near the minimum of each curve indicates the number of basis 
functions at that minimum.
\label{fig:Heuhf}}
\end{figure}

We now consider results for the Hartree-Fock approximation, since this tests both the diagonal approximation and the completeness of the hybrid basis. Moreover, computations at the Hartree-Fock level allows us to test larger-$Z$ atoms without too much additional method development compared to many-electron correlated calculations. The previous hybrid gausslet study stopped at $Z=2$, and we have found that going beyond $Z=4$ with a coordinate-product coordinate transformation is difficult. Here we show results up to Ne, with all calculations performed on a desktop, and going beyond Ne would not be particularly difficult.

In Fig. \ref{fig:Heuhf} we show Hartree-Fock results for the helium atom in a hybrid basis.
The basis was chosen for its highly accurate nuclear cusp treatment, but it also has excellent completeness for any S function, including the HF orbital. The key approximation in these calculations is the diagonal representation of the interactions, which becomes extremely accurate for larger values of $n_s$; sub-micro-Hartree errors are easily reached with $n_s=9$.

In Fig. \ref{fig:Cuhf} we show unrestricted Hartree-Fock results for the carbon atom in a hybrid basis. In this case we used two different Gaussian bases to try to assess any errors due to inadequate treatment of the nuclear cusp by the Gaussian basis, but the differences in the results were only slightly noticeable, and only at $n_s=7$. Our final energy is accurate to about $10^{-5}$ Hartree.

\begin{figure}[t]
\includegraphics[width=0.8\columnwidth]{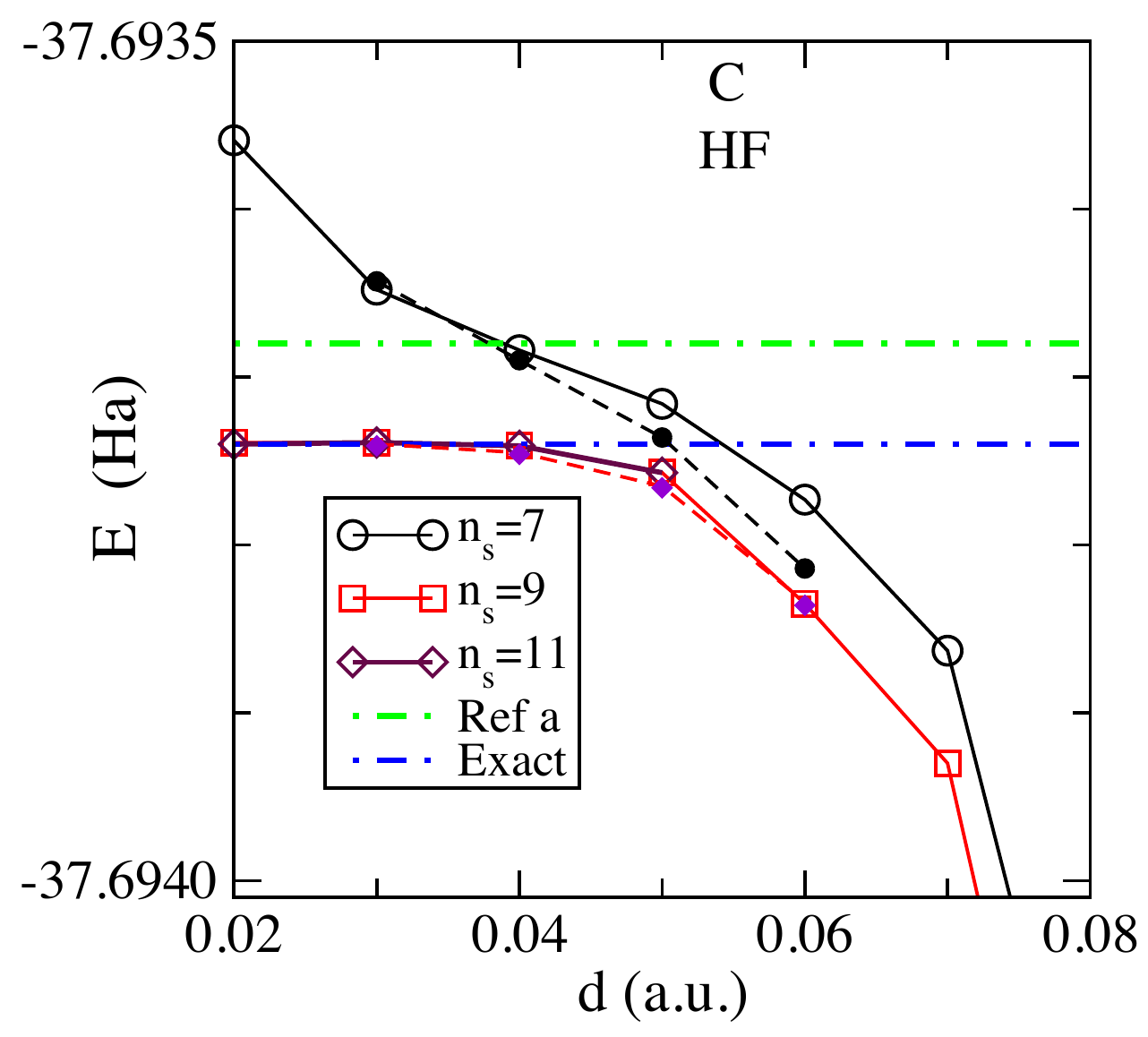}
    \caption{
Unrestricted Hartree-Fock results for a carbon atom in a hybrid PGDG-gausslet/GTO basis versus the minimum gausslet spacing $d$ (in a.u.) for box size $R_b=16$ a.u.  The open symbols use the S and P functions of a cc-pV6Z  basis\cite{Cccpv6z}, while the corresponding small closed symbols use a AHGBS-9 basis\cite{AHGBS9}, which only has S and P functions. The green dot-dashed line (Ref. a) shows previous results from~\cite{cook_unrestricted_1981}, and the black dot-dashed line is a numerically exact result from MRchem~\cite{mrchem}: $-37.6937404 $ Ha. Our final total energy for $n_s=11$ of $-37.69374(1) $ Ha is consistent with the exact result. The number of basis functions for $n_s=7, d=0.05$ was 2623; for $n_s=9, d=0.04$, 4688; for $n_s=11, d=0.03$, 8650.
\label{fig:Cuhf}}
\end{figure}

In Fig. \ref{fig:Neuhf} we show Hartree-Fock results for a neon atom, far beyond what would have been possible with previous gausslet bases. In this case we compare with the highly accurate numerical results of Cinal~\cite{cinal_highly_2020}. Here, one needs $n_s=9$ at least to reach sub-milli-Hartree total energies. At $n_s=11$, we find an energy of $-128.54708(1)$ Ha for $d\sim 0.007-0.01$. This is off by about $2\times10^{-5}$ Ha from Cinal's result of $-128.54709810938$. 
\begin{figure}[t]
\includegraphics[width=0.7\columnwidth]{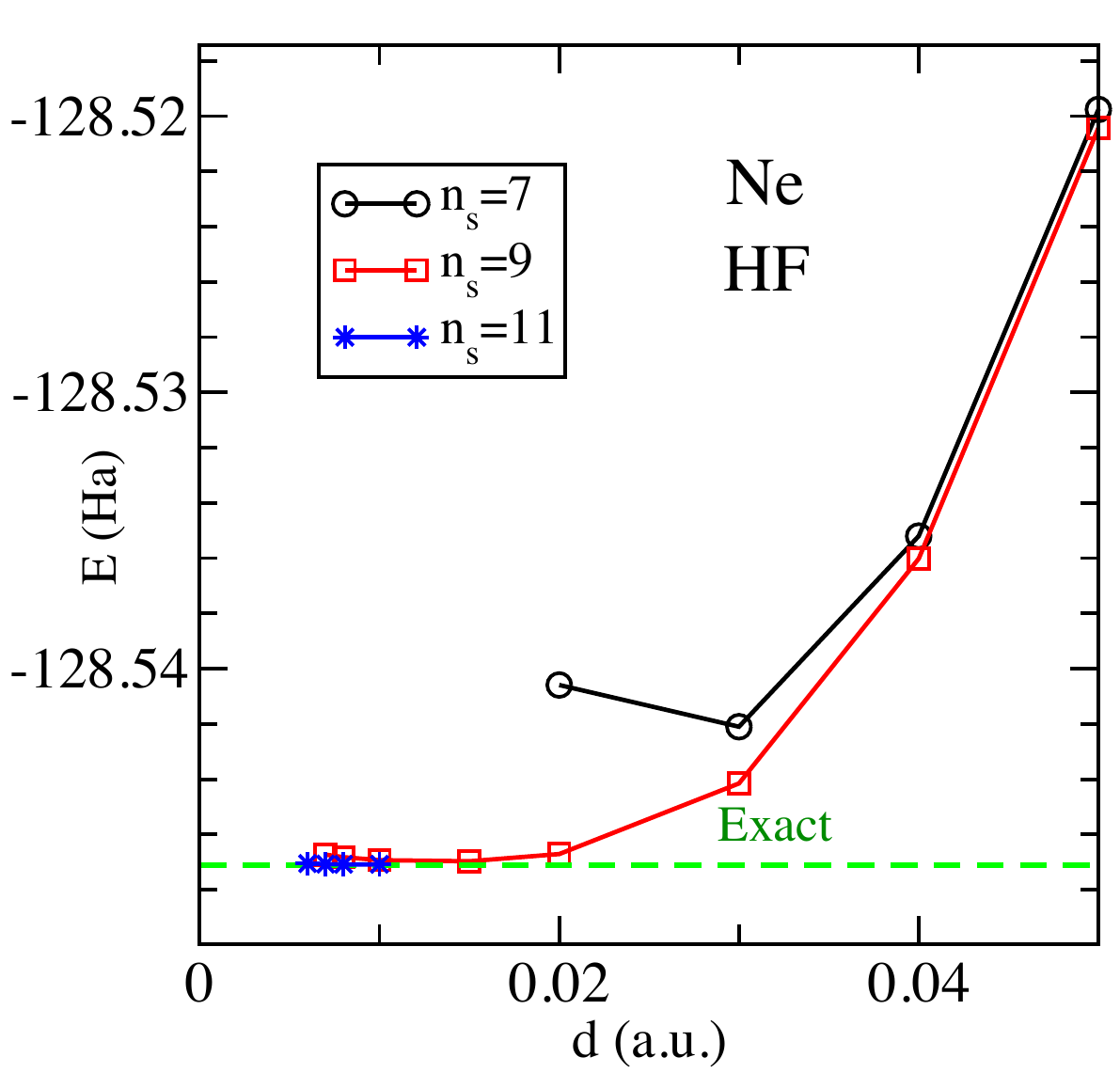}

    \caption{
Restricted Hartree-Fock total energies for a Neon atom  with a hybrid PGDG-gausslet/GTO basis versus the minimum gausslet spacing $d$ (in a.u.) for box size $R_b=15$ a.u.  The GTO utilizes the S  functions of a cc-pV6Z  basis\cite{Cccpv6z} The dashed line shows the numerical results of Cinal~\cite{cinal_highly_2020}. The number of basis functions for $n_s=7, d=0.03$ was 2748; for $n_s=9, d=0.02$, 5754; for $n_s=11, d=0.03$, 12776.
\label{fig:Neuhf}}
\end{figure}

The question of whether the beryllium dimer has a bound state was a long-standing challenge for quantum chemistry. Results from DMRG with transcorrelated basis sets have given convincing results that there is a bound state~\cite{sharma_spectroscopic_2014}, in disagreement with HF results exhibiting a purely repulsive potential~\cite{roeggen_interatomic_1996}.
For the Be atom, Ivanov found a lower-symmetry lower-energy UHF state~\cite{ivanov_hartree-fock_1998}, with energy $E=-14.57336$  Ha, versus the RHF result of $E=-14.57302$ Ha. Our tests on the Be atom verify these results: for both $n_s=9$ and $n_s=11$ we find $E=-14.573351$ Ha for the UHF energy, with a hybrid PGDG-gausslet/cc-pV6Z basis\cite{Beccpv6z}, using S and P GTOs.  For our UHF calculations we specify  that the number of $\alpha$ and $\beta$ electrons are identical ($S^z=0)$; otherwise, the orbitals are unconstrained, and the system does not have to have a definite state of total spin $S$.
In testing our NG bases on Be$_2$, we discovered that there is a UHF solution that is substantially lower than the RHF solution, which provides for a bound state of the dimer. For example, at separation $R=4$, we find  a binding energy of $0.01070$ Ha for $n_s=9$, $d=0.05$, and $0.01073$ Ha for $n_s=7$, $d=0.075$, using UHF energies for both the molecule and atom. This lower energy UHF solution was subsequently verified through a stability analysis using the PYSCF package~\cite{pyscf} and a large traditional GTO basis~\cite{Sharma_private}. In Fig. \ref{fig:Be2} we show the results for the potential energy of the system as a function of $R$, for both RHF and UHF.  (In the RHF calculations, we used the atomic RHF result.)  The resulting potential energy 
curve seems to be qualitatively similar to the exact binding curve, but the exact binding energy has a minimum near $R=4.75$ a.u. instead of $R=4.25$ a.u., and its minimum is only about $4.3$ mH\cite{sharma_spectroscopic_2014} versus $12.2$ mH from UHF.

\begin{figure}[t]
\includegraphics[width=0.8\columnwidth]{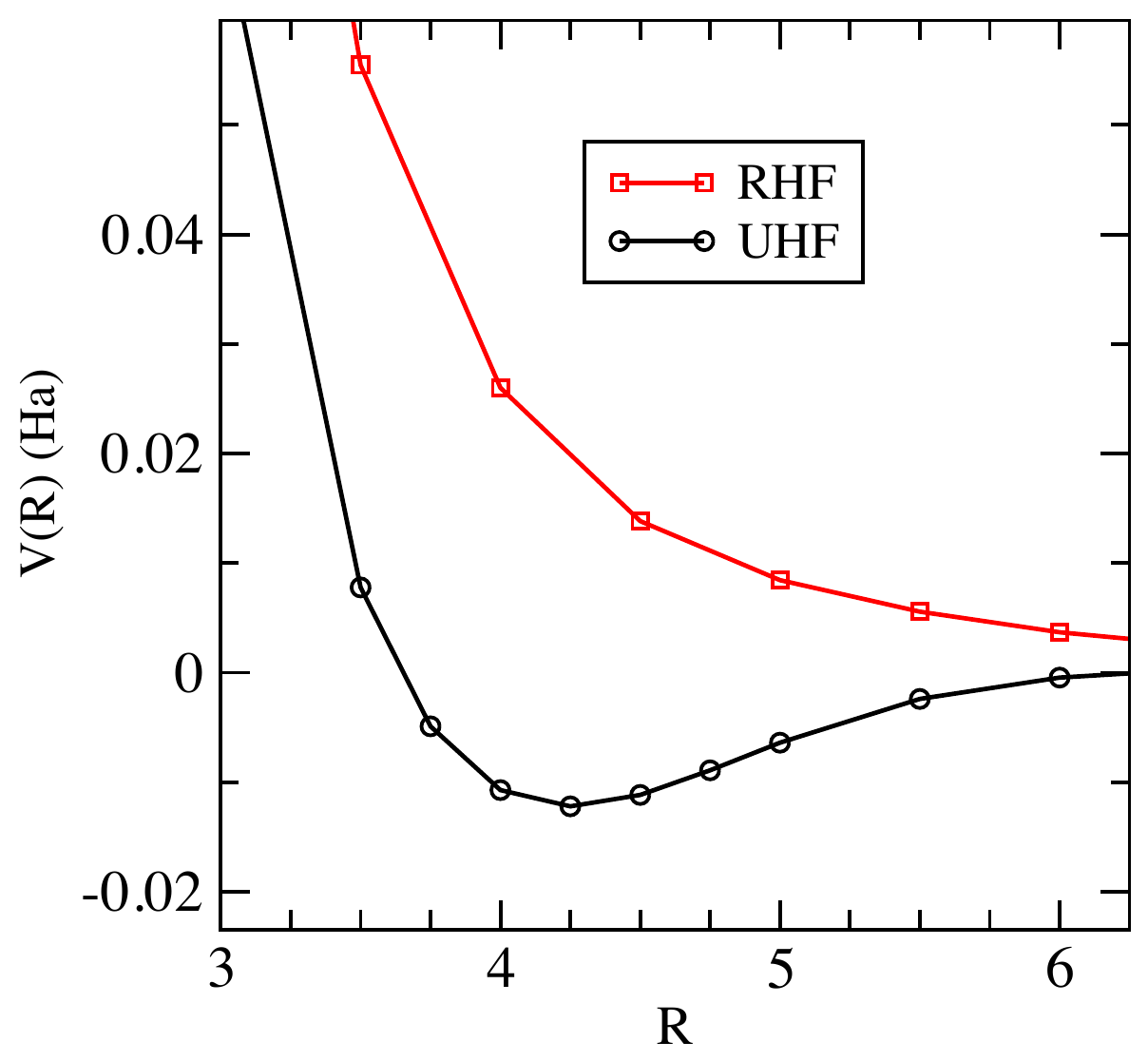}
    \caption{
Potential energy for RHF and UHF for the beryllium dimer within NG bases at separation $R$.  The NG basis was a hybrid PGDG-gausslet/GTO basis utilizing the S and P functions of a cc-pV6Z  basis\cite{Beccpv6z}, with $n_s=7$ and $d=0.075$.
\label{fig:Be2}}
\end{figure}

\begin{figure}[t]
\includegraphics[width=0.8\columnwidth]{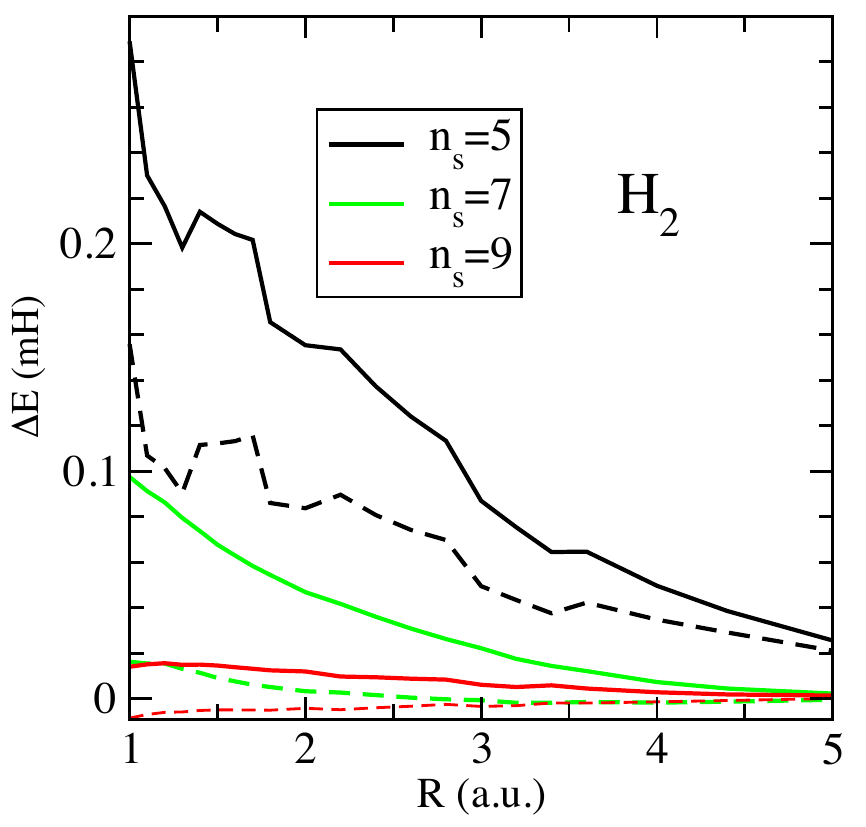}
    \caption{
Exact diagonalization within NG bases for the H$_2$ molecule at separation $R$, relative to the exact results of Ref.~\cite{kolos_new_1986}.  The NG basis was a hybrid PGDG-gausslet/GTO basis utilizing the S and P functions of a cc-pV6Z  basis\cite{Hccpv6z} For $n_s=5$, $7$, and $9$, we used $d =0.4$, $0.2$, and $0.1$, respectively. The solid lines are colored according to $n_s$ as in the legend; the dashed lines of the same color are the results corrected by the double-occupancy cusp correction.
\label{fig:H2pot}}
\end{figure}

Turning now to fully correlated calculations, in Fig. \ref{fig:H2pot} we show full-CI  complete basis results for the hydrogen molecule as a function of nuclear separation. Even at $n_s=5$, the errors are well below chemical accuracy. At large separation, errors are near the $\mu H$ level. Here the two electrons are usually on different atoms so the basis set errors associated with the electron-electron cusp are less important.  The errors are larger at the shortest distances, but still of order $10^{-5}$Ha at $R=1$ for $n_s=9$.  We also show energies corrected by the double occupancy electron-electron cusp correction.  Here, the cusp correction is not as effective as it was for the non-nested hybrid gausslet bases of~\cite{qiu_hybrid_2021}, although the cusp correction does always decrease the errors. It is possible that the effectiveness of the correction is reduced because the fitting of the coefficients for small molecules is sensitive to the difference between non-nested and nested gausslets.

\section{Conclusions}
In this paper we have described nested gausslet (NG) bases, including  pure Gaussian distorted gausslet (PGDG) bases, as well as hybrid versions of these bases with ordinary Gaussians. All gausslet bases are localized and orthonormal from the start.  Previous gausslet bases were not able to effectively treat atoms with atomic number $Z$ greater than 3 or 4; NG bases are able to treat much larger $Z$. NG bases are in general more efficient than previous gausslet bases, since there are no regions away from nuclei where the spacing between functions is inappropriately small. 

The PGDG basis allows for analytic computation of all the integrals needed to construct the Hamiltonian, once one has decomposed the Coulomb interaction into a modest sum of Gaussians. The computation time to construct the Hamiltonian matrices is quite modest, and the diagonal approximation allows bases with at least 10,000 functions to be used on a desktop. 

In constructing these bases we have developed new insight into one-dimensional diagonal basis sets. We have proved an important theorem relating completeness, orthogonality, zero-moment properties of the functions, and diagonalization of the coordinate operator ($x$) matrix. This mathematical development greatly eases the construction of NG bases.

We have demonstrated NG bases with tests on small systems with a focus on very high accuracy, allowing for clear-cut benchmarking.  One can rather easily achieve very small basis set errors at the Hartree-Fock level, typically around $10^{-5}$ Ha. At the fully correlated level, we expect electron-electron cusp errors to be somewhat larger, but in H$_2$ these errors are still around $10^{-5}$ Ha or less. However, the reduction in correlation errors from a recently developed electron-electron cusp correction are not as impressive as in previous work. 

A motivation for this work has been to improve scaling for DMRG calculations in quantum chemistry. The diagonal form of the electron-electron interaction translates immediately to better scaling in DMRG calculations, and the locality of the basis is also
expected to reduce entanglement, but we leave tests using DMRG (and other correlated methods beyond two electrons) for later work.  We also leave lower-accuracy but larger-system tests for later work.

\begin{acknowledgments}
We thank Yiheng Qiu, Sandeep Sharma, and Garnet Chan for helpful conversations.  This work was funded by NSF through Grant DMR-2110041 (SRW). 
\end{acknowledgments}



\bibliography{NestedGaussletRefs}

\begin{thebibliography}{31}%
\makeatletter
\providecommand \@ifxundefined [1]{%
 \@ifx{#1\undefined}
}%
\providecommand \@ifnum [1]{%
 \ifnum #1\expandafter \@firstoftwo
 \else \expandafter \@secondoftwo
 \fi
}%
\providecommand \@ifx [1]{%
 \ifx #1\expandafter \@firstoftwo
 \else \expandafter \@secondoftwo
 \fi
}%
\providecommand \natexlab [1]{#1}%
\providecommand \enquote  [1]{``#1''}%
\providecommand \bibnamefont  [1]{#1}%
\providecommand \bibfnamefont [1]{#1}%
\providecommand \citenamefont [1]{#1}%
\providecommand \href@noop [0]{\@secondoftwo}%
\providecommand \href [0]{\begingroup \@sanitize@url \@href}%
\providecommand \@href[1]{\@@startlink{#1}\@@href}%
\providecommand \@@href[1]{\endgroup#1\@@endlink}%
\providecommand \@sanitize@url [0]{\catcode `\\12\catcode `\$12\catcode
  `\&12\catcode `\#12\catcode `\^12\catcode `\_12\catcode `\%12\relax}%
\providecommand \@@startlink[1]{}%
\providecommand \@@endlink[0]{}%
\providecommand \url  [0]{\begingroup\@sanitize@url \@url }%
\providecommand \@url [1]{\endgroup\@href {#1}{\urlprefix }}%
\providecommand \urlprefix  [0]{URL }%
\providecommand \Eprint [0]{\href }%
\providecommand \doibase [0]{http://dx.doi.org/}%
\providecommand \selectlanguage [0]{\@gobble}%
\providecommand \bibinfo  [0]{\@secondoftwo}%
\providecommand \bibfield  [0]{\@secondoftwo}%
\providecommand \translation [1]{[#1]}%
\providecommand \BibitemOpen [0]{}%
\providecommand \bibitemStop [0]{}%
\providecommand \bibitemNoStop [0]{.\EOS\space}%
\providecommand \EOS [0]{\spacefactor3000\relax}%
\providecommand \BibitemShut  [1]{\csname bibitem#1\endcsname}%
\let\auto@bib@innerbib\@empty
\bibitem [{\citenamefont {Dunlap}(2000)}]{dunlap_robust_2000}%
  \BibitemOpen
  \bibfield  {author} {\bibinfo {author} {\bibfnamefont {B.~I.}\ \bibnamefont
  {Dunlap}},\ }\href {\doibase 10.1039/B000027M} {\bibfield  {journal}
  {\bibinfo  {journal} {Physical Chemistry Chemical Physics}\ }\textbf
  {\bibinfo {volume} {2}},\ \bibinfo {pages} {2113} (\bibinfo {year} {2000})},\
  \bibinfo {note} {publisher: The Royal Society of Chemistry}\BibitemShut
  {NoStop}%
\bibitem [{\citenamefont {Hohenstein}\ \emph {et~al.}(2012)\citenamefont
  {Hohenstein}, \citenamefont {Parrish},\ and\ \citenamefont
  {Mart{\'\i}nez}}]{HohensteinEtAl2012_I}%
  \BibitemOpen
  \bibfield  {author} {\bibinfo {author} {\bibfnamefont {E.~G.}\ \bibnamefont
  {Hohenstein}}, \bibinfo {author} {\bibfnamefont {R.~M.}\ \bibnamefont
  {Parrish}}, \ and\ \bibinfo {author} {\bibfnamefont {T.~J.}\ \bibnamefont
  {Mart{\'\i}nez}},\ }\href@noop {} {\bibfield  {journal} {\bibinfo  {journal}
  {The Journal of Chemical Physics}\ }\textbf {\bibinfo {volume} {137}},\
  \bibinfo {pages} {044103} (\bibinfo {year} {2012})}\BibitemShut {NoStop}%
\bibitem [{\citenamefont {Parrish}\ \emph {et~al.}(2012)\citenamefont
  {Parrish}, \citenamefont {Hohenstein}, \citenamefont {Mart{\'\i}nez},\ and\
  \citenamefont {Sherrill}}]{HohensteinEtAl2012_II}%
  \BibitemOpen
  \bibfield  {author} {\bibinfo {author} {\bibfnamefont {R.~M.}\ \bibnamefont
  {Parrish}}, \bibinfo {author} {\bibfnamefont {E.~G.}\ \bibnamefont
  {Hohenstein}}, \bibinfo {author} {\bibfnamefont {T.~J.}\ \bibnamefont
  {Mart{\'\i}nez}}, \ and\ \bibinfo {author} {\bibfnamefont {C.~D.}\
  \bibnamefont {Sherrill}},\ }\href@noop {} {\bibfield  {journal} {\bibinfo
  {journal} {The Journal of Chemical Physics}\ }\textbf {\bibinfo {volume}
  {137}},\ \bibinfo {pages} {224106} (\bibinfo {year} {2012})}\BibitemShut
  {NoStop}%
\bibitem [{\citenamefont {Parrish}\ \emph {et~al.}(2013)\citenamefont
  {Parrish}, \citenamefont {Hohenstein}, \citenamefont {Schunck}, \citenamefont
  {Sherrill},\ and\ \citenamefont {Mart\'{\i}nez}}]{ParrishEtAl2013}%
  \BibitemOpen
  \bibfield  {author} {\bibinfo {author} {\bibfnamefont {R.~M.}\ \bibnamefont
  {Parrish}}, \bibinfo {author} {\bibfnamefont {E.~G.}\ \bibnamefont
  {Hohenstein}}, \bibinfo {author} {\bibfnamefont {N.~F.}\ \bibnamefont
  {Schunck}}, \bibinfo {author} {\bibfnamefont {C.~D.}\ \bibnamefont
  {Sherrill}}, \ and\ \bibinfo {author} {\bibfnamefont {T.~J.}\ \bibnamefont
  {Mart\'{\i}nez}},\ }\href {\doibase 10.1103/PhysRevLett.111.132505}
  {\bibfield  {journal} {\bibinfo  {journal} {Phys. Rev. Lett.}\ }\textbf
  {\bibinfo {volume} {111}},\ \bibinfo {pages} {132505} (\bibinfo {year}
  {2013})}\BibitemShut {NoStop}%
\bibitem [{\citenamefont {Lu}\ and\ \citenamefont {Ying}(2015)}]{LuYing2015}%
  \BibitemOpen
  \bibfield  {author} {\bibinfo {author} {\bibfnamefont {J.}~\bibnamefont
  {Lu}}\ and\ \bibinfo {author} {\bibfnamefont {L.}~\bibnamefont {Ying}},\
  }\href@noop {} {\bibfield  {journal} {\bibinfo  {journal} {Journal of
  Computational Physics}\ }\textbf {\bibinfo {volume} {302}},\ \bibinfo {pages}
  {329} (\bibinfo {year} {2015})}\BibitemShut {NoStop}%
\bibitem [{\citenamefont {Aquilante}\ \emph {et~al.}(2007)\citenamefont
  {Aquilante}, \citenamefont {Pedersen},\ and\ \citenamefont
  {Lindh}}]{aquilante_low-cost_2007}%
  \BibitemOpen
  \bibfield  {author} {\bibinfo {author} {\bibfnamefont {F.}~\bibnamefont
  {Aquilante}}, \bibinfo {author} {\bibfnamefont {T.~B.}\ \bibnamefont
  {Pedersen}}, \ and\ \bibinfo {author} {\bibfnamefont {R.}~\bibnamefont
  {Lindh}},\ }\href {\doibase 10.1063/1.2736701} {\bibfield  {journal}
  {\bibinfo  {journal} {The Journal of Chemical Physics}\ }\textbf {\bibinfo
  {volume} {126}},\ \bibinfo {pages} {194106} (\bibinfo {year}
  {2007})}\BibitemShut {NoStop}%
\bibitem [{\citenamefont {Gygi}\ and\ \citenamefont
  {Galli}(1995)}]{gygi_real-space_1995}%
  \BibitemOpen
  \bibfield  {author} {\bibinfo {author} {\bibfnamefont {F.}~\bibnamefont
  {Gygi}}\ and\ \bibinfo {author} {\bibfnamefont {G.}~\bibnamefont {Galli}},\
  }\href {\doibase 10.1103/PhysRevB.52.R2229} {\bibfield  {journal} {\bibinfo
  {journal} {Physical Review B}\ }\textbf {\bibinfo {volume} {52}},\ \bibinfo
  {pages} {R2229} (\bibinfo {year} {1995})},\ \bibinfo {note} {publisher:
  American Physical Society}\BibitemShut {NoStop}%
\bibitem [{\citenamefont {Jones}\ \emph {et~al.}(2016)\citenamefont {Jones},
  \citenamefont {Rouet}, \citenamefont {Lawler}, \citenamefont {Vecharynski},
  \citenamefont {Ibrahim}, \citenamefont {Williams}, \citenamefont {Abeln},
  \citenamefont {Yang}, \citenamefont {McCurdy}, \citenamefont {Haxton},
  \citenamefont {Li},\ and\ \citenamefont {Rescigno}}]{jones_efficient_2016}%
  \BibitemOpen
  \bibfield  {author} {\bibinfo {author} {\bibfnamefont {J.~R.}\ \bibnamefont
  {Jones}}, \bibinfo {author} {\bibfnamefont {F.-H.}\ \bibnamefont {Rouet}},
  \bibinfo {author} {\bibfnamefont {K.~V.}\ \bibnamefont {Lawler}}, \bibinfo
  {author} {\bibfnamefont {E.}~\bibnamefont {Vecharynski}}, \bibinfo {author}
  {\bibfnamefont {K.~Z.}\ \bibnamefont {Ibrahim}}, \bibinfo {author}
  {\bibfnamefont {S.}~\bibnamefont {Williams}}, \bibinfo {author}
  {\bibfnamefont {B.}~\bibnamefont {Abeln}}, \bibinfo {author} {\bibfnamefont
  {C.}~\bibnamefont {Yang}}, \bibinfo {author} {\bibfnamefont {W.}~\bibnamefont
  {McCurdy}}, \bibinfo {author} {\bibfnamefont {D.~J.}\ \bibnamefont {Haxton}},
  \bibinfo {author} {\bibfnamefont {X.~S.}\ \bibnamefont {Li}}, \ and\ \bibinfo
  {author} {\bibfnamefont {T.~N.}\ \bibnamefont {Rescigno}},\ }\href@noop {}
  {\bibfield  {journal} {\bibinfo  {journal} {Molecular Physics}\ }\textbf
  {\bibinfo {volume} {114}} (\bibinfo {year} {2016})}\BibitemShut {NoStop}%
\bibitem [{\citenamefont {White}(2017)}]{white_hybrid_2017}%
  \BibitemOpen
  \bibfield  {author} {\bibinfo {author} {\bibfnamefont {S.~R.}\ \bibnamefont
  {White}},\ }\href {\doibase 10.1063/1.5007066} {\bibfield  {journal}
  {\bibinfo  {journal} {The Journal of Chemical Physics}\ }\textbf {\bibinfo
  {volume} {147}},\ \bibinfo {pages} {244102} (\bibinfo {year}
  {2017})}\BibitemShut {NoStop}%
\bibitem [{\citenamefont {White}\ and\ \citenamefont
  {Stoudenmire}(2019)}]{white_multisliced_2019}%
  \BibitemOpen
  \bibfield  {author} {\bibinfo {author} {\bibfnamefont {S.~R.}\ \bibnamefont
  {White}}\ and\ \bibinfo {author} {\bibfnamefont {E.~M.}\ \bibnamefont
  {Stoudenmire}},\ }\href {\doibase 10.1103/PhysRevB.99.081110} {\bibfield
  {journal} {\bibinfo  {journal} {Physical Review B}\ }\textbf {\bibinfo
  {volume} {99}},\ \bibinfo {pages} {081110} (\bibinfo {year}
  {2019})}\BibitemShut {NoStop}%
\bibitem [{\citenamefont {Qiu}\ and\ \citenamefont
  {White}(2021)}]{qiu_hybrid_2021}%
  \BibitemOpen
  \bibfield  {author} {\bibinfo {author} {\bibfnamefont {Y.}~\bibnamefont
  {Qiu}}\ and\ \bibinfo {author} {\bibfnamefont {S.~R.}\ \bibnamefont
  {White}},\ }\href {\doibase 10.1063/5.0068887} {\bibfield  {journal}
  {\bibinfo  {journal} {The Journal of Chemical Physics}\ }\textbf {\bibinfo
  {volume} {155}},\ \bibinfo {pages} {184107} (\bibinfo {year}
  {2021})}\BibitemShut {NoStop}%
\bibitem [{\citenamefont {Light}\ and\ \citenamefont
  {Carrington~Jr.}(2000)}]{light_discrete-variable_2000}%
  \BibitemOpen
  \bibfield  {author} {\bibinfo {author} {\bibfnamefont {J.~C.}\ \bibnamefont
  {Light}}\ and\ \bibinfo {author} {\bibfnamefont {T.}~\bibnamefont
  {Carrington~Jr.}},\ }in\ \href {\doibase
  https://doi.org/10.1002/9780470141731.ch4} {\emph {\bibinfo {booktitle}
  {Advances in {Chemical} {Physics}}}}\ (\bibinfo  {publisher} {John Wiley \&
  Sons, Ltd},\ \bibinfo {year} {2000})\ pp.\ \bibinfo {pages}
  {263--310}\BibitemShut {NoStop}%
\bibitem [{\citenamefont {Evenbly}\ and\ \citenamefont
  {White}(2018)}]{evenbly_representation_2018}%
  \BibitemOpen
  \bibfield  {author} {\bibinfo {author} {\bibfnamefont {G.}~\bibnamefont
  {Evenbly}}\ and\ \bibinfo {author} {\bibfnamefont {S.~R.}\ \bibnamefont
  {White}},\ }\href {\doibase 10.1103/PhysRevA.97.052314} {\bibfield  {journal}
  {\bibinfo  {journal} {Physical Review A}\ }\textbf {\bibinfo {volume} {97}},\
  \bibinfo {pages} {052314} (\bibinfo {year} {2018})}\BibitemShut {NoStop}%
\bibitem [{\citenamefont {Scott}\ \emph {et~al.}(2006)\citenamefont {Scott},
  \citenamefont {Aubert-Frecon},\ and\ \citenamefont
  {Grotendorst}}]{scott_new_2006}%
  \BibitemOpen
  \bibfield  {author} {\bibinfo {author} {\bibfnamefont {T.~C.}\ \bibnamefont
  {Scott}}, \bibinfo {author} {\bibfnamefont {M.}~\bibnamefont
  {Aubert-Frecon}}, \ and\ \bibinfo {author} {\bibfnamefont {J.}~\bibnamefont
  {Grotendorst}},\ }\href {\doibase 10.1016/j.chemphys.2005.10.031} {\bibfield
  {journal} {\bibinfo  {journal} {Chemical Physics}\ }\textbf {\bibinfo
  {volume} {324}},\ \bibinfo {pages} {323} (\bibinfo {year}
  {2006})}\BibitemShut {NoStop}%
\bibitem [{\citenamefont {Woon}\ and\ \citenamefont {Dunning}()}]{Heccpv6z}%
  \BibitemOpen
  \bibfield  {author} {\bibinfo {author} {\bibfnamefont {D.~E.}\ \bibnamefont
  {Woon}}\ and\ \bibinfo {author} {\bibfnamefont {T.~H.}\ \bibnamefont
  {Dunning}, \bibfnamefont {Jr}},\ }\href@noop {} {\enquote {\bibinfo {title}
  {unpublished},}\ }\bibinfo {note} {As referenced in 'van Mourik et al, Mol
  Phys, 96, 529-547 (1999)' (as reference 48)}\BibitemShut {NoStop}%
\bibitem [{\citenamefont {Peterson}\ \emph {et~al.}(1994)\citenamefont
  {Peterson}, \citenamefont {Woon},\ and\ \citenamefont {Dunning}}]{Hccpv6z}%
  \BibitemOpen
  \bibfield  {author} {\bibinfo {author} {\bibfnamefont {K.~A.}\ \bibnamefont
  {Peterson}}, \bibinfo {author} {\bibfnamefont {D.~E.}\ \bibnamefont {Woon}},
  \ and\ \bibinfo {author} {\bibfnamefont {T.~H.}\ \bibnamefont {Dunning},
  \bibfnamefont {Jr}},\ }\href {\doibase 10.1063/1.466884} {\bibfield
  {journal} {\bibinfo  {journal} {J. Chem. Phys.}\ }\textbf {\bibinfo {volume}
  {100}},\ \bibinfo {pages} {7410} (\bibinfo {year} {1994})}\BibitemShut
  {NoStop}%
\bibitem [{\citenamefont {Cinal}(2020)}]{cinal_highly_2020}%
  \BibitemOpen
  \bibfield  {author} {\bibinfo {author} {\bibfnamefont {M.}~\bibnamefont
  {Cinal}},\ }\href {\doibase 10.1007/s10910-020-01144-z} {\bibfield  {journal}
  {\bibinfo  {journal} {Journal of Mathematical Chemistry}\ }\textbf {\bibinfo
  {volume} {58}},\ \bibinfo {pages} {1571} (\bibinfo {year}
  {2020})}\BibitemShut {NoStop}%
\bibitem [{\citenamefont {Wilson}\ \emph {et~al.}(1996)\citenamefont {Wilson},
  \citenamefont {van Mourik},\ and\ \citenamefont {Dunning}}]{Cccpv6z}%
  \BibitemOpen
  \bibfield  {author} {\bibinfo {author} {\bibfnamefont {A.~K.}\ \bibnamefont
  {Wilson}}, \bibinfo {author} {\bibfnamefont {T.}~\bibnamefont {van Mourik}},
  \ and\ \bibinfo {author} {\bibfnamefont {T.~H.}\ \bibnamefont {Dunning}},\
  }\href {\doibase 10.1016/s0166-1280(96)80048-0} {\bibfield  {journal}
  {\bibinfo  {journal} {J. Mol. Struc-THEOCHEM}\ }\textbf {\bibinfo {volume}
  {388}},\ \bibinfo {pages} {339} (\bibinfo {year} {1996})}\BibitemShut
  {NoStop}%
\bibitem [{\citenamefont {Lehtola}(2020)}]{AHGBS9}%
  \BibitemOpen
  \bibfield  {author} {\bibinfo {author} {\bibfnamefont {S.}~\bibnamefont
  {Lehtola}},\ }\href {\doibase 10.1063/1.5144964} {\bibfield  {journal}
  {\bibinfo  {journal} {J. Chem. Phys.}\ }\textbf {\bibinfo {volume} {152}},\
  \bibinfo {pages} {134108} (\bibinfo {year} {2020})}\BibitemShut {NoStop}%
\bibitem [{\citenamefont {Cook}(1981)}]{cook_unrestricted_1981}%
  \BibitemOpen
  \bibfield  {author} {\bibinfo {author} {\bibfnamefont {D.~B.}\ \bibnamefont
  {Cook}},\ }\href {\doibase 10.1007/BF00550429} {\bibfield  {journal}
  {\bibinfo  {journal} {Theoretica chimica acta}\ }\textbf {\bibinfo {volume}
  {58}},\ \bibinfo {pages} {155} (\bibinfo {year} {1981})}\BibitemShut
  {NoStop}%
\bibitem [{\citenamefont {Wind}\ \emph {et~al.}(2023)\citenamefont {Wind},
  \citenamefont {Bj{\o}rgve}, \citenamefont {Brakestad}, \citenamefont
  {Gerez~S.}, \citenamefont {Jensen}, \citenamefont {Eik{\aa}s},\ and\
  \citenamefont {Frediani}}]{mrchem}%
  \BibitemOpen
  \bibfield  {author} {\bibinfo {author} {\bibfnamefont {P.}~\bibnamefont
  {Wind}}, \bibinfo {author} {\bibfnamefont {M.}~\bibnamefont {Bj{\o}rgve}},
  \bibinfo {author} {\bibfnamefont {A.}~\bibnamefont {Brakestad}}, \bibinfo
  {author} {\bibfnamefont {G.~A.}\ \bibnamefont {Gerez~S.}}, \bibinfo {author}
  {\bibfnamefont {S.~R.}\ \bibnamefont {Jensen}}, \bibinfo {author}
  {\bibfnamefont {R.~D.~R.}\ \bibnamefont {Eik{\aa}s}}, \ and\ \bibinfo
  {author} {\bibfnamefont {L.}~\bibnamefont {Frediani}},\ }\href@noop {}
  {\bibfield  {journal} {\bibinfo  {journal} {Journal of Chemical Theory and
  Computation}\ }\textbf {\bibinfo {volume} {19}},\ \bibinfo {pages} {137}
  (\bibinfo {year} {2023})}\BibitemShut {NoStop}%
\bibitem [{\citenamefont {Sharma}\ \emph {et~al.}(2014)\citenamefont {Sharma},
  \citenamefont {Yanai}, \citenamefont {Booth}, \citenamefont {Umrigar},\ and\
  \citenamefont {Chan}}]{sharma_spectroscopic_2014}%
  \BibitemOpen
  \bibfield  {author} {\bibinfo {author} {\bibfnamefont {S.}~\bibnamefont
  {Sharma}}, \bibinfo {author} {\bibfnamefont {T.}~\bibnamefont {Yanai}},
  \bibinfo {author} {\bibfnamefont {G.~H.}\ \bibnamefont {Booth}}, \bibinfo
  {author} {\bibfnamefont {C.~J.}\ \bibnamefont {Umrigar}}, \ and\ \bibinfo
  {author} {\bibfnamefont {G.~K.-L.}\ \bibnamefont {Chan}},\ }\href@noop {}
  {\bibfield  {journal} {\bibinfo  {journal} {The Journal of Chemical Physics}\
  }\textbf {\bibinfo {volume} {140}},\ \bibinfo {pages} {104112} (\bibinfo
  {year} {2014})}\BibitemShut {NoStop}%
\bibitem [{\citenamefont {Røeggen}\ and\ \citenamefont
  {Almlöf}(1996)}]{roeggen_interatomic_1996}%
  \BibitemOpen
  \bibfield  {author} {\bibinfo {author} {\bibfnamefont {I.}~\bibnamefont
  {Røeggen}}\ and\ \bibinfo {author} {\bibfnamefont {J.}~\bibnamefont
  {Almlöf}},\ }\href {\doibase
  10.1002/(SICI)1097-461X(1996)60:1<453::AID-QUA44>3.0.CO;2-A} {\bibfield
  {journal} {\bibinfo  {journal} {International Journal of Quantum Chemistry}\
  }\textbf {\bibinfo {volume} {60}},\ \bibinfo {pages} {453} (\bibinfo {year}
  {1996})}\BibitemShut {NoStop}%
\bibitem [{\citenamefont {Ivanov}(1998)}]{ivanov_hartree-fock_1998}%
  \BibitemOpen
  \bibfield  {author} {\bibinfo {author} {\bibfnamefont {M.~V.}\ \bibnamefont
  {Ivanov}},\ }\href@noop {} {\bibfield  {journal} {\bibinfo  {journal}
  {Physics Letters A}\ }\textbf {\bibinfo {volume} {239}},\ \bibinfo {pages}
  {72} (\bibinfo {year} {1998})}\BibitemShut {NoStop}%
\bibitem [{\citenamefont {Prascher}\ \emph {et~al.}(2011)\citenamefont
  {Prascher}, \citenamefont {Woon}, \citenamefont {Peterson}, \citenamefont
  {Dunning},\ and\ \citenamefont {Wilson}}]{Beccpv6z}%
  \BibitemOpen
  \bibfield  {author} {\bibinfo {author} {\bibfnamefont {B.~P.}\ \bibnamefont
  {Prascher}}, \bibinfo {author} {\bibfnamefont {D.~E.}\ \bibnamefont {Woon}},
  \bibinfo {author} {\bibfnamefont {K.~A.}\ \bibnamefont {Peterson}}, \bibinfo
  {author} {\bibfnamefont {T.~H.}\ \bibnamefont {Dunning}}, \ and\ \bibinfo
  {author} {\bibfnamefont {A.~K.}\ \bibnamefont {Wilson}},\ }\href {\doibase
  10.1007/s00214-010-0764-0} {\bibfield  {journal} {\bibinfo  {journal} {Theor.
  Chem. Acc.}\ }\textbf {\bibinfo {volume} {128}},\ \bibinfo {pages} {69}
  (\bibinfo {year} {2011})}\BibitemShut {NoStop}%
\bibitem [{\citenamefont {Sun}\ \emph {et~al.}(2020)\citenamefont {Sun},
  \citenamefont {Zhang}, \citenamefont {Banerjee}, \citenamefont {Bao},
  \citenamefont {Barbry}, \citenamefont {Blunt}, \citenamefont {Bogdanov},
  \citenamefont {Booth}, \citenamefont {Chen}, \citenamefont {Cui},
  \citenamefont {Eriksen}, \citenamefont {Gao}, \citenamefont {Guo},
  \citenamefont {Hermann}, \citenamefont {Hermes}, \citenamefont {Koh},
  \citenamefont {Koval}, \citenamefont {Lehtola}, \citenamefont {Li},
  \citenamefont {Liu}, \citenamefont {Mardirossian}, \citenamefont {McClain},
  \citenamefont {Motta}, \citenamefont {Mussard}, \citenamefont {Pham},
  \citenamefont {Pulkin}, \citenamefont {Purwanto}, \citenamefont {Robinson},
  \citenamefont {Ronca}, \citenamefont {Sayfutyarova}, \citenamefont
  {Scheurer}, \citenamefont {Schurkus}, \citenamefont {Smith}, \citenamefont
  {Sun}, \citenamefont {Sun}, \citenamefont {Upadhyay}, \citenamefont {Wagner},
  \citenamefont {Wang}, \citenamefont {White}, \citenamefont {Whitfield},
  \citenamefont {Williamson}, \citenamefont {Wouters}, \citenamefont {Yang},
  \citenamefont {Yu}, \citenamefont {Zhu}, \citenamefont {Berkelbach},
  \citenamefont {Sharma}, \citenamefont {Sokolov},\ and\ \citenamefont
  {Chan}}]{pyscf}%
  \BibitemOpen
  \bibfield  {author} {\bibinfo {author} {\bibfnamefont {Q.}~\bibnamefont
  {Sun}}, \bibinfo {author} {\bibfnamefont {X.}~\bibnamefont {Zhang}}, \bibinfo
  {author} {\bibfnamefont {S.}~\bibnamefont {Banerjee}}, \bibinfo {author}
  {\bibfnamefont {P.}~\bibnamefont {Bao}}, \bibinfo {author} {\bibfnamefont
  {M.}~\bibnamefont {Barbry}}, \bibinfo {author} {\bibfnamefont {N.~S.}\
  \bibnamefont {Blunt}}, \bibinfo {author} {\bibfnamefont {N.~A.}\ \bibnamefont
  {Bogdanov}}, \bibinfo {author} {\bibfnamefont {G.~H.}\ \bibnamefont {Booth}},
  \bibinfo {author} {\bibfnamefont {J.}~\bibnamefont {Chen}}, \bibinfo {author}
  {\bibfnamefont {Z.-H.}\ \bibnamefont {Cui}}, \bibinfo {author} {\bibfnamefont
  {J.~J.}\ \bibnamefont {Eriksen}}, \bibinfo {author} {\bibfnamefont
  {Y.}~\bibnamefont {Gao}}, \bibinfo {author} {\bibfnamefont {S.}~\bibnamefont
  {Guo}}, \bibinfo {author} {\bibfnamefont {J.}~\bibnamefont {Hermann}},
  \bibinfo {author} {\bibfnamefont {M.~R.}\ \bibnamefont {Hermes}}, \bibinfo
  {author} {\bibfnamefont {K.}~\bibnamefont {Koh}}, \bibinfo {author}
  {\bibfnamefont {P.}~\bibnamefont {Koval}}, \bibinfo {author} {\bibfnamefont
  {S.}~\bibnamefont {Lehtola}}, \bibinfo {author} {\bibfnamefont
  {Z.}~\bibnamefont {Li}}, \bibinfo {author} {\bibfnamefont {J.}~\bibnamefont
  {Liu}}, \bibinfo {author} {\bibfnamefont {N.}~\bibnamefont {Mardirossian}},
  \bibinfo {author} {\bibfnamefont {J.~D.}\ \bibnamefont {McClain}}, \bibinfo
  {author} {\bibfnamefont {M.}~\bibnamefont {Motta}}, \bibinfo {author}
  {\bibfnamefont {B.}~\bibnamefont {Mussard}}, \bibinfo {author} {\bibfnamefont
  {H.~Q.}\ \bibnamefont {Pham}}, \bibinfo {author} {\bibfnamefont
  {A.}~\bibnamefont {Pulkin}}, \bibinfo {author} {\bibfnamefont
  {W.}~\bibnamefont {Purwanto}}, \bibinfo {author} {\bibfnamefont {P.~J.}\
  \bibnamefont {Robinson}}, \bibinfo {author} {\bibfnamefont {E.}~\bibnamefont
  {Ronca}}, \bibinfo {author} {\bibfnamefont {E.~R.}\ \bibnamefont
  {Sayfutyarova}}, \bibinfo {author} {\bibfnamefont {M.}~\bibnamefont
  {Scheurer}}, \bibinfo {author} {\bibfnamefont {H.~F.}\ \bibnamefont
  {Schurkus}}, \bibinfo {author} {\bibfnamefont {J.~E.~T.}\ \bibnamefont
  {Smith}}, \bibinfo {author} {\bibfnamefont {C.}~\bibnamefont {Sun}}, \bibinfo
  {author} {\bibfnamefont {S.-N.}\ \bibnamefont {Sun}}, \bibinfo {author}
  {\bibfnamefont {S.}~\bibnamefont {Upadhyay}}, \bibinfo {author}
  {\bibfnamefont {L.~K.}\ \bibnamefont {Wagner}}, \bibinfo {author}
  {\bibfnamefont {X.}~\bibnamefont {Wang}}, \bibinfo {author} {\bibfnamefont
  {A.}~\bibnamefont {White}}, \bibinfo {author} {\bibfnamefont {J.~D.}\
  \bibnamefont {Whitfield}}, \bibinfo {author} {\bibfnamefont {M.~J.}\
  \bibnamefont {Williamson}}, \bibinfo {author} {\bibfnamefont
  {S.}~\bibnamefont {Wouters}}, \bibinfo {author} {\bibfnamefont
  {J.}~\bibnamefont {Yang}}, \bibinfo {author} {\bibfnamefont {J.~M.}\
  \bibnamefont {Yu}}, \bibinfo {author} {\bibfnamefont {T.}~\bibnamefont
  {Zhu}}, \bibinfo {author} {\bibfnamefont {T.~C.}\ \bibnamefont {Berkelbach}},
  \bibinfo {author} {\bibfnamefont {S.}~\bibnamefont {Sharma}}, \bibinfo
  {author} {\bibfnamefont {A.~Y.}\ \bibnamefont {Sokolov}}, \ and\ \bibinfo
  {author} {\bibfnamefont {G.~K.-L.}\ \bibnamefont {Chan}},\ }\href {\doibase
  10.1063/5.0006074} {\bibfield  {journal} {\bibinfo  {journal} {The Journal of
  Chemical Physics}\ }\textbf {\bibinfo {volume} {153}},\ \bibinfo {pages}
  {024109} (\bibinfo {year} {2020})}\BibitemShut {NoStop}%
\bibitem [{\citenamefont {Sharma}()}]{Sharma_private}%
  \BibitemOpen
  \bibfield  {author} {\bibinfo {author} {\bibfnamefont {S.}~\bibnamefont
  {Sharma}},\ }\href@noop {} {}\bibinfo {note} {Private communication. The
  basis was ANO-RCC, P.-O. Widmark, P.-Å. Malmqvist and B. O. Roos, Theor.
  Chim. Acta 77, 291-306 (1990).}\BibitemShut {Stop}%
\bibitem [{\citenamefont {Kol/os}\ \emph {et~al.}(1986)\citenamefont {Kol/os},
  \citenamefont {Szalewicz},\ and\ \citenamefont {Monkhorst}}]{kolos_new_1986}%
  \BibitemOpen
  \bibfield  {author} {\bibinfo {author} {\bibfnamefont {W.}~\bibnamefont
  {Kol/os}}, \bibinfo {author} {\bibfnamefont {K.}~\bibnamefont {Szalewicz}}, \
  and\ \bibinfo {author} {\bibfnamefont {H.~J.}\ \bibnamefont {Monkhorst}},\
  }\href {\doibase 10.1063/1.450258} {\bibfield  {journal} {\bibinfo  {journal}
  {The Journal of Chemical Physics}\ }\textbf {\bibinfo {volume} {84}},\
  \bibinfo {pages} {3278} (\bibinfo {year} {1986})}\BibitemShut {NoStop}%
\bibitem [{\citenamefont {Beylkin}\ and\ \citenamefont
  {Monz{\'o}n}(2005)}]{BeylkinMonzon2005}%
  \BibitemOpen
  \bibfield  {author} {\bibinfo {author} {\bibfnamefont {G.}~\bibnamefont
  {Beylkin}}\ and\ \bibinfo {author} {\bibfnamefont {L.}~\bibnamefont
  {Monz{\'o}n}},\ }\href@noop {} {\bibfield  {journal} {\bibinfo  {journal}
  {Applied and Computational Harmonic Analysis}\ }\textbf {\bibinfo {volume}
  {19}},\ \bibinfo {pages} {17} (\bibinfo {year} {2005})}\BibitemShut {NoStop}%
\bibitem [{\citenamefont {Beylkin}\ and\ \citenamefont
  {Monz{\'o}n}(2010)}]{BeylkinMonzon2010}%
  \BibitemOpen
  \bibfield  {author} {\bibinfo {author} {\bibfnamefont {G.}~\bibnamefont
  {Beylkin}}\ and\ \bibinfo {author} {\bibfnamefont {L.}~\bibnamefont
  {Monz{\'o}n}},\ }\href@noop {} {\bibfield  {journal} {\bibinfo  {journal}
  {Applied and Computational Harmonic Analysis}\ }\textbf {\bibinfo {volume}
  {28}},\ \bibinfo {pages} {131} (\bibinfo {year} {2010})}\BibitemShut
  {NoStop}%
\bibitem [{\citenamefont {Daubechies}(1992)}]{DaubechiesBook}%
  \BibitemOpen
  \bibfield  {author} {\bibinfo {author} {\bibfnamefont {I.}~\bibnamefont
  {Daubechies}},\ }\href@noop {} {\emph {\bibinfo {title} {Ten Lectures on
  Wavelets}}}\ (\bibinfo  {publisher} {Society for Industrial and Applied
  Mathematics},\ \bibinfo {year} {1992})\BibitemShut {NoStop}%
\end{thebibliography}%
\section*{Appendix A: Coordinate Mappings}
With a coordinate transformation, a gausslet $G(x)$ is replaced by $\sqrt{u'(x)} G(u(x))$. The density of basis functions is given by $\rho(x) \equiv u'(x)$; the spacing between basis functions is approximately $1/\rho(x)$. Here we give coordinate mappings for the simple case of one atom centered at $x_0$. 

The sinh mapping is defined by a core size $a$ and scale parameter $s$, and is given by
\begin{eqnarray}
   \rho(x) &=& \left(s \sqrt{(x-x_0)^2+a^2}\right)^{-1}\nonumber\\
   u(x) &=& \sinh^{-1}[(x-x_0)/a]/s \\
   x(u) &=& x_0 + a  \sinh(su) .\nonumber
\end{eqnarray}
The function spacing at the origin is $as$. Typical parameters are $s = 0.7$ and $a=s/Z$, where $Z$ is the atomic number. For this mapping, the basis function spacing increases without bound away from the origin. This fact is usually inconvenient for numerical integration, so typically a constant shift $1/w$ is added to $\rho$, yielding an asymptotic spacing of $w$. Our standard choice is $w=10$. 

Another useful mapping, which we call the erf/x mapping, is given by
\begin{equation}
   \rho(x)= \frac{n
   \left(\text{erf}\left(\frac{x}{c}
   \right)-\text{erf}\left(\frac{x}{
   d}\right)\right)}{2 x \ln
   \left(\frac{d}{c}\right)}.
\end{equation}
This mapping is similar to the sinh mapping, with a core size controlled by $c$, falling roughly as $\sim 1/x$ for $c<x<d$, but dropping off to zero quickly for $|x|>d$. The total integral of $\rho(x)$ is $n$, so this density contributes $n$ functions when added to another density. 

\section*{Appendix B: Coulomb Integrals}
We use a long-known trick for performing all three dimensional Coulomb integrals, representing $1/r$ as a sum of Gaussians via discretization of the integral
\begin{equation}
    \frac{1}{r}  = \frac{1}{\sqrt{\pi}} \int_{-\infty}^\infty dt\ 
    e^{-r^2 t^2} .
\end{equation}
Then the 3D integrals separate into products of $x$, $y$, and $z$ integrals because $e^{-\zeta r^2} = e^{-\zeta x^2}e^{-\zeta y^2}e^{-\zeta z^2}$.

We can use the sinh mapping of Appendix A to construct a suitable discretization of this integral. (Alternative approaches can be found in~\cite{BeylkinMonzon2005,BeylkinMonzon2010}.) 
Let $u(t)$ be a sinh mapping with $x_0=0$; we find that the parameters $s=0.3$, $a=0.03$ are close to ideal for moderate accuracy calculations. Then change variables $t\to u$, and discretize uniformly in $u$,
using a symmetric grid with the origin at the midpoint of two grid points.  The integral is even, so only need to take into account the positive grid points, with contribution adjusted by a factor of 2. Concretely, we take
 $u_i = (i-1/2)$ for $i=1\ldots M$, where $M=45$, and then define $t_i = t(u_i)$. 
 This choice yields the approximation
\begin{equation}
    \frac{1}{r}  \approx \sum_i c_i e^{-\zeta_i r^2}
\end{equation}
with
\begin{eqnarray}
    \zeta_i &=& t_i^2 \nonumber \\
    c_i &=& \frac{2 \Delta}{\sqrt{\pi} \rho(t_i)}.
\end{eqnarray}
For $r$ ranging from $10^{-3}$ to $10$, this approximation is accurate to 7 or 8 digits.

If greater accuracy or range is desired, then the parameters
$s=0.16$, $c=0.01$, and $M=115$ yield a relative error of about $10^{-13}$ for $r$ ranging from $10^{-5}$ to $100$.  Indeed, we have used the latter high-accuracy choice of parameters for the calculations in this paper.

\section*{Appendix C: Additional measures of optimality for potential diagonal bases}
Properties not guaranteed by COMX are locality and a self-integral that is not nearly zero
$w_i$ (see Eq. (\ref{eq:mom})); in fact, COMX allows for $w_i=0$,
in which case the function does not act like a $\delta$-function at all.

One scale-invariant measure tied to the size of $w_i$ is the positivity 
\begin{equation}
P = \frac{\int_x S(x)}{\int_x |S(x)|}  ,
\label{eq:positivity}
\end{equation}
which one would like to maximize.
Completeness and orthogonality limit the positivity to well below 1; gausslets with 8th and 10th order completeness have $P=0.693$ and $P=0.675$, respectively, which seem to be nearly optimal values for these level of completeness. In contrast, the
sinc function has $P=0$ due to the fact that it is not absolutely integrable. Meanwhile, the Meyer scaling function~\cite{DaubechiesBook} defines another orthogonal diagonal basis. It is defined in terms of its Fourier transform $\hat{S}$ as 
\begin{equation}
    \hat{S}(\omega)=\begin{cases}
    \frac{1}{\sqrt{2\pi}}, & \vert\omega\vert\leq2\pi/3\\
    0, & \vert\omega\vert\geq4\pi/3\\
    \frac{1}{\sqrt{2\pi}}\cos\left(\frac{\pi}{2}\nu\left(\frac{3\vert\omega\vert}{2\pi}-1\right)\right), & \text{otherwise},
    \end{cases}
\end{equation}
where $\nu$ continuosly interpolates $\nu(0)=0$ to $\nu(1)=1$. The most common choice for $\nu$ producing well-localized $S$ is a seventh-order polynomial which achieves the derivative conditions $\nu^{(k)}(0) = \nu^{(k)}(1) = 0$ for $k=1,2,3$. This choice yields a value of $P=0.585$.

Another scale-invariant quantification of the quality of a basis function is based on the usual $\Delta p \, \Delta x$ measure of uncertainty in a wavepacket. The {\it uncertainty} is
\begin{equation}
U = 4\left[\int_x S'(x)^2\right] \times \left[\int_x S(x)^2 (x-\bar x)^2\right]
\end{equation}
where $\bar x=\int_x S(x)^2$. The factor of 4 ensures that a Gaussian, the minimum-uncertainty wavepacket, has $U=1$. In general, smaller $U$ is better. Gausslets of order 8 and 10 have $U=2.11$ and $2.30$, respectively. Meanwhile, sincs have infinite uncertainty, and Meyer scaling functions have $U=4.09$. 


\end{document}